\documentclass[journal,10pt]{IEEEtran}
\IEEEoverridecommandlockouts

\newif\ifarxiv\arxivtrue

\usepackage[utf8]{inputenc}
\usepackage{color}
\usepackage{amsmath,amssymb,mathtools,dsfont,bm,amsthm}
\usepackage{algorithm,algorithmic}

\usepackage{cite}
\usepackage{siunitx}
\usepackage{soul}
\usepackage{subcaption}
\usepackage{subfloat}

\newtheorem{remark}{Remark}
\newtheorem{lemma}{Lemma}
\newtheorem{theorem}{Theorem}

\let\OLDthebibliography\thebibliography
\renewcommand\thebibliography[1]{
	\OLDthebibliography{#1}
	\setlength{\parskip}{-.4pt}
}

\DeclarePairedDelimiter{\norm}{\lVert}{\rVert}

\title{Comeback Kid: Resilience for Mixed-Critical Wireless Network Resource Management}
\author{\IEEEauthorblockN{Robert-Jeron Reifert, Stefan Roth, Alaa Alameer Ahmad and Aydin Sezgin}\\
\IEEEauthorblockA{Institute of Digital Communication Systems, Ruhr University Bochum, Bochum, Germany\\
Email: \{robert-.reifert,stefan.roth-k21,alaa.alameerahmad,aydin.sezgin\}@rub.de}
\thanks{Part of this paper was presented at the IEEE International Conference on Communications Workshops, May 2022 \cite{reifert2022energy}. 
\ifarxiv%
This work has been submitted to the IEEE for possible publication. Copyright may be transferred without notice, after which this version may no longer be accessible.\newline
\else%
An extended version of this paper is available on arxiv \cite{comebackkid}. \newline
\fi%
This work was funded in part by the Federal Ministry of Education and Research (BMBF) of the Federal Republic of Germany (F\"orderkennzeichen 01IS18063A, ReMiX), and in part by the Deutsche Forschungsgemeinschaft (DFG, German Research Foundation) under Germany's Excellence Strategy - EXC 2092 CASA - 390781972}
}
\date{December 2021}

\begin{document}

\maketitle

\begin{abstract}
    The future sixth generation (6G) of communication systems is envisioned to provide numerous applications in safety-critical contexts, e.g., driverless traffic, modular industry, and smart cities, which require outstanding performance, high reliability and fault tolerance, as well as autonomy. 
    Ensuring criticality awareness for diverse functional safety applications and providing fault tolerance in an autonomous manner are essential for future 6G systems. Therefore, this paper proposes jointly employing the concepts of resilience and mixed criticality.
    In this work, we conduct 
    physical layer resource management in cloud-based networks under the rate-splitting paradigm, which is a promising factor towards achieving high resilience. We recapitulate the concepts individually, outline a joint metric to measure the criticality-aware resilience, and verify its merits in a case study.
    We, thereby, formulate a non-convex optimization problem, derive an efficient iterative algorithm, propose four resilience mechanisms differing in quality and time of adaption, and conduct extensive numerical simulations.
    Towards this end, we propose a highly autonomous rate-splitting-enabled physical layer resource management algorithm for future 6G networks respecting mixed-critical quality of service (QoS) levels and providing high levels of resilience.
    Results emphasize the considerable improvements of incorporating a mixed criticality-aware resilience strategy under channel outages and strict QoS demands. The rate-splitting paradigm is particularly shown to overcome state-of-the-art interference management techniques, and the resilience and throughput adaption over consecutive outage events reveals the proposed schemes contribution towards enabling future 6G networks.

\end{abstract}

\section{Introduction}

\subsection{Motivation}
    The road towards the sixth generation (6G) of wireless communication networks is already being pursued by researchers around the globe \cite{8922617,8820755}. Through a wide range of applications, the empowerment of anytime anywhere access, and an overwhelming amount of connected devices, 6G brings enormous challenges towards the development of future network technologies. Particularly, use cases such as wireless-based cloud office for small and home office, smart cities, and smart factory, depend on high-performance and reliable networks \cite{erricson}.\\ 
    \indent
	It is forecasted that the number of internet of things (IoT) connections increases from $14.6$ billion in 2021 to $30.2$ billion in 2027 \cite{erricson}.
	Hence, there is a huge number of devices with different levels of criticality, such as safety-critcal, mission-critical, and low-critical IoT devices, which also coexist within one system \cite{sota20}.
	In industrial context, mixed criticality corresponds to different priorities 
	of applications, e.g., a security monitoring system is more critical than a maintenance scheduler. To ensure fulfilling the quality of service (QoS) demands, we investigate QoS target capabilities of the considered network. On the physical layer, QoS is often translated to the allocated rates of network participants \cite{9145363,7938594}. Thereby, the QoS assigned to the nodes is designed to match the desired data rates (target rates), which depend on the subscribed contract (service provider networks) or criticality level (industrial context).\\ 
	\indent
	The continuously increasing IoT connectivity brings along another hurdle in the design of future 6G networks, namely resilience. Resilience is the capacity of a system to absorb a disturbance and reorganize while undergoing change so as to still retain its essential function. In other words, resilience captures a system's ability to maintain functionality facing errors, adapt to erroneous influences, and recover the functionality in a timely manner. 
	As this should be achieved in an automated fashion without the need for human interaction \cite[Principle P16]{STERBENZ20101245}, autonomy is a central principle at the design stage.
	Resilience is a topic of interest throughout various areas of industry and academia (psychology \cite{doi:10.1177/0959354318783584}, industrial-ecological systems \cite{5936913}, communications \cite{sota7}, security \cite{6584933}). Generally, the overall concept of resilience includes different system characteristics, i.e., \emph{detection}, \emph{remediation}, and \emph{recovery} \cite{STERBENZ20101245}. \\
	\indent
	Especially in the context of high critical IoT nodes, the concept of resilience is of significant importance.
	Providing a robust communication system is equally vital that timely recovery from failures to an acceptable service level, in order to deploy such systems for real-time safety-critical applications. In this context, \textit{rate-splitting} comes into play as an efficient and robust communication scheme \cite{9445019,7513415}. Originating in the early 80's \cite{1056307}, and shown to achieve within one-bit of the interference channel capacity \cite{4675741}, rate-splitting mutliple access (RSMA) achieved significant attention in research, e.g., \cite{5910112,8846706,mao2022ratesplitting,8385499,6034007,7996346,6169200}. %
	\begin{figure*}
		\centering
		\includegraphics[width=.9\linewidth]{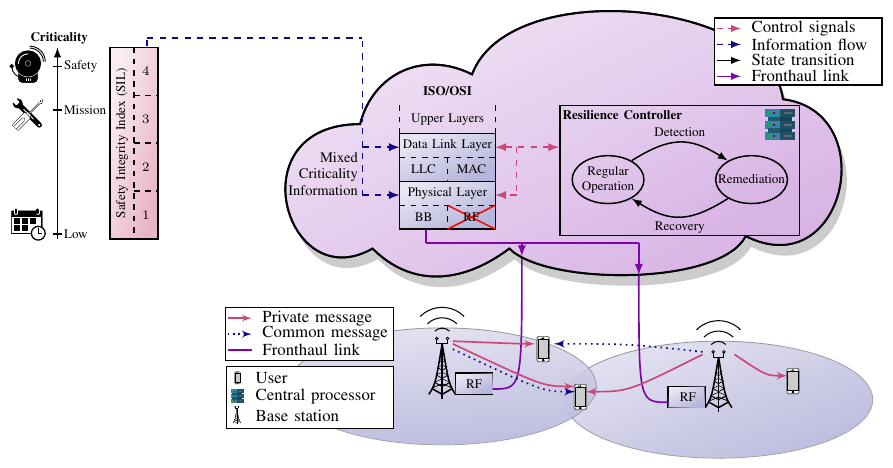}\vspace*{-.3cm}
		\caption{The considered cloud-based network architecture including mixed criticality, resilience, and RSMA elements.\protect\footnotemark}
		\label{overall_figure}
		\vspace*{-.6cm}
	\end{figure*}%
	Conventionally, each user is assigned to a single message stream (\emph{private message}), however, under RSMA, an additional \emph{common message} is utilized for two reasons: $(1)$ 
	Taking over parts of the data transmission from the private message
	, and $(2)$ interference mitigation by common message decoding at other users to reduce the interference level when decoding their own private messages. Especially due to the enhanced opportunities for communication (multiple message streams, \emph{redundancy}) and different purposes of streams (private and common, \emph{diversity}), we identify RSMA as a promising resiliency enhancing paradigm.\\
	\indent 
	In this work, we aim at tackling the fundamental challenge of integrating mixed criticality levels in the physical layer of a wireless communication system. Thereby designing a resilient RSMA-enabled network architecture ensuring high robustness, automated adaption, and fast recovery, especially when the network resources are constrained. An example of such architecture can be seen in Fig.~\ref{overall_figure}, where we distinguish between mixed criticality information, the ISO/OSI lower layers, and a resilience controller at the cloud connected to the RSMA-enabled radio access network. Mixed criticality information serves as input to the ISO/OSI model, while the resilience controller manages the network's operation via control signals.
	A resilient resource management promises to enhance future networks' performance in an automated manner without inducing major losses in terms of service quality. To the best of the authors knowledge, this is the first work which considers resilience and mixed criticality for the physical layer of wireless communication systems under the RSMA paradigm. We aim to provide the general concepts, design recommendations, and a case study to evaluate the proposed methodology.

\subsection{Related Literature}
    The considered methodology of combining resilience and mixed criticality for physical layer resource management in this paper is related to works residing in the domains of resilience (especially network resilience), robustness, reliability, and mixed criticality of communication systems. Additionally, there are recent related works considering the RSMA paradigm.

	\footnotetext{The safety integrity levels (SILs) are applications of mixed criticality in real systems (IEC 61508) \cite{sil}. For more details we refer to section \ref{ssec:mixedcrit}.}%
	The term resilience has its roots in the Latin verb \emph{resiliere}, meaning to rebound or recoil \cite{mcaslan2010concept}, e.g., "\emph{[...] saepe super ripam stagni consistere, saepe in gelidos resilire lacus [...]}", taken from ancient literature about a metamorphosis from human to frogs, which often sit at waterside and hastily jump back (recoil) into the cold water \cite{resilience_ancient}. Nowadays, resilience relates to many engineering fields \cite{5936913,hosseini2016review}, environmental and regional studies \cite{holm2013}, psychology \cite{doi:10.1177/0959354318783584}, as well as economics \cite{10.1093/ser/mww015}. While definitions and methodologies may differ among the research directions, a common overlap is that resilience refers to some kind of disruption and the return to the normal situation \cite{Najarian2019}. 
	A great amount of work towards (network) resilience in communication systems was done by the ResiliNets initiative \cite{resilienets}. Especially, the seminal paper \cite{STERBENZ20101245} point out axioms, strategies, and principles of resilience in communication networks, with a focus on the internet. 
	The recent book \cite{rak2020} describes techniques for disaster-resilient communication networks and includes many works of international researchers. 
	Work \cite{borhani2020secure} elaborates on communications in industrial IoT describing potential network architectures, with a focus on security. 
	In \cite{sota1}, the authors investigate the reliability of the IP multimedia subsystem and define the interplay between \emph{availability} and \emph{reliability} and their relation to resilience. 
	Work \cite{sota3} describes a framework to evaluate network dependability and performability in the face of challenges such as attacks and disasters. 
	The authors point out, that redundancy and diversity increase the reliability but also increase the costs. The Survivable Mobile Wireless Networking (SUMOWIN) project \cite{10.1145/570681.570685} reports challenges and ideas for routing and networking in unreliable mobile wireless networks, emphasizing the need for adaptive and agile networking and sattelite support. 
	As many related works on resilience in communications regard networking problems, the work \cite{STERBENZ20101245} notes a major challenge as \emph{failures at a lower layer}, e.g., a fibre cut causes a link-layer failure, which has to be remediated by re-routing at a higher layer. However, with the utilization of wireless communications under cloud-based architectures, it is necessary to include resilience at the lower layers to assist the overall network resiliency capabilities.
	
	In this work, we adapt the concept of resilience for the physical layer of wireless communication systems, a field in which only limited considerations of resilience exist. 6G communication is the enabling infrastructure for many critical applications and therefore resilience become a very important topic in these scenarios. Some related chapters in \cite{rak2020} reside in this domain: 
	%
    (a) In \cite{Bruzgiene2020}, QoS in modular positioning systems, wireless sensor networks, and free-space optical (FSO) communication systems is reviewed under weather disruptions.
    (b) In \cite{Ivanov2020}, availability of FSO systems under atmospheric impacts is studied.
    (c) In \cite{Cinkler2020}, resilience enhancing techniques for 5G systems are studied, namely frequency fallback, segment interleaving, and multi-operator protection.
    Further works on resilience for wireless communications include \cite{7306543}, with radio and FSO backhauling for network resilience, \cite{sota5}, emphasizing the need for intelligent fault management and mitigation strategies at design and run-time, \cite{sota6}, using unlicensed spectrum band as well as non-terrestrial networks, \cite{sota10}, where an experimental PC-USRP hardware platform evaluated the performance in the face of pulse interference for in industrial environments, and \cite{sota12}, where reliability and availability of FSO indoor systems is studied.
	Most of these works consider resilient systems, but, do not consider holistic metrics for resilience.
	
	The concepts of robustness and reliability, being one aspect of resilience, correspond, in part, to the concept of ultra-reliable low-latency communication \cite{sota13,sota14}. Especially in \cite{sota13}, the characteristics and dynamics of wireless channels under low-latency and high-reliability constraints is considered. 
	It is stated that wireless channels are unreliable due to fading, which can be tackled by introducing diversity techniques, i.e., time/frequency/spatial-diversity, sub-carrier coding, and multiple antennas. However, such systems typically do not consider the entirety of resilience aspects and include only a singe criticality level.
	
	The concept of mixed criticality has been introduced in 2007 for task scheduling in real-time systems \cite{sota17}. Since then, mixed criticality has been adapted for a wider range of applications in the field of communications \cite{sota18,sota19,sota20,sota21}.
	In general, mixed criticality refers to communication links having different priority levels, typically because failures have different consequences. Usually, the links are categorized as safety, mission, and low-critical links. AirTight, a protocol for time-critical cyber physical systems (CPS) including real-time and mixed criticality requirements, has been proposed in \cite{sota18} and \cite{sota19}. Thereby, a criticality-aware fault model is utilized to capture external interference, in which mixed criticality is implemented as different amounts of maximum deadline misses of traffic links. The goal of reliable real-time performance for data delivery in industrial wireless sensor networks was approached in \cite{sota20}. Thereby, a criticality-aware wireless fieldbus protocol was proposed for different data flows with different importance levels, i.e., delay and reliability constraints. In the same context, work \cite{sota21} states that high-critical real-time flows must be guaranteed reliability in the face of errors. 
	However, these related works lack considerations of the lower layers on the protocol stack, especially the physical layer of wireless communication networks.

	Smart interference management schemes, such as RSMA, inherently posses some resilience characteristics, i.e., robustness, as they cope well with high interference scenarios, and can guarantee a certain connectivity. 
	We give a brief overview of related works on robustness of RSMA and QoS-aware RSMA schemes. A general overview of RSMA is provided in \cite{mao2022ratesplitting}, revisiting fundamental concepts and future trends, see also references therein. Scalability and robustness, especially robustness against channel state information (CSI) imperfections, of a RSMA-enabled network is analyzed in \cite{9445019}. Work \cite{7513415} considers the robust design problem for achieving max-min fairness among all users and shows the superior performance of the proposed RSMA scheme under CSI uncertainty. QoS constrainted power minimization is conducted in \cite{7513415,9145363}.
	Therein, the authors investigate the minimization of weighted-sum of transmit powers subject to per-user QoS constraints. Hence, such scheme is only feasible in networks, where the QoS is achievable. 
	However, this assumption is rather optimistic and we herein propose a more generalized scheme, which provides feasible solutions for networks with insufficient resources.
	
	While the related literature includes several general research directions, to fill some of the gaps, we build up onto these ideas and extend the contribution towards considering mixed criticality on the physical layer and the combination of resilience and mixed criticality for wireless communication resource management, especially using RSMA.
	
\subsection{Contributions}
    In this paper, we design a general framework for wireless communication systems that accounts for the merits of mixed criticality in the physical layer, and also provides aspects of resilience, i.e., high reliability, automated adaption to failures, and rapid recovery. As such, we recap the individual concepts of resilience and mixed criticality and define their manifestations for the physical layer resource management. Thereby, we strike a connection of QoS and criticality level, and define mathematical formulations of resilience metrics for wireless communications. Combining those two concepts, we employ the unified framework in a case study, involving RSMA-enabled cloud-radio access networks (C-RAN). 
    Thus, this work is one step forward to design criticality-aware, highly-automated 6G communication systems for various industrial and service applications. The contributions of the work are summarized as
    \begin{itemize}
        \item Establishing the concepts of resilience and mixed criticality for cloud-assisted wireless networks to enable the vision of 6G 
        \item Designing a unified framework to consider resilience and mixed criticality jointly in physical layer communications resource management
        \item Presenting a case study on RSMA-enabled cloud networks, in which the following resilience mechanisms are employed: 
        (1) A \emph{rate adaption} mechanism adjusts the QoS to a value feasible with the selected parameters,
        (2) a \emph{beamformer adaption} mechanism is optimizing the beamformers according to the new situation,
        (3) a \emph{BS-user-allocation adaption} mechanism optimizes the allocation of BSs to users and
        (4) a \emph{comprehensive adaption} mechanism re-visits the original formulations and adjusts all parameters jointly.\vspace{-.2cm}
    \end{itemize}

\subsection{Notation and Organization}

Vectors $\bm{a}$ (matrices $\bm{A}$) are denoted as bold lower-case (upper-case) letters, respectively. Sets $\mathcal{A}$ are denoted in calligraphic and have cardinality $|\mathcal{A}|$. $\mathbb{C}$ denotes the complex field, $\bm{0}_N$ an all-zero vector of dimension $N\times 1$, $|\cdot|$ the absolute value, $\norm[]{\cdot}_p$ the $L_p$-norm, and $(\cdot)^H$ the Hermitian transpose operator. At last, $\mathrm{Re}\{\cdot\}$ is the real part of a complex number.

\ifarxiv
    The rest of this paper is organized as follows: Section \ref{sec:resmixcrit} introduces the concepts of resilience and mixed criticality individually and subsequently provides a joint metric combining those concepts for the physical layer resource management. Then, in section \ref{sec:casestudy} we conduct a case study on RSMA-enabled C-RAN as follows: \ref{ssec:moti} case study motivation; \ref{ssec:sysmod} system model introduction; \ref{ssec:opt} problem formulation and solution; \ref{ssec:resmixcritalloc} resilient and criticality-aware resource allocation algorithm; \ref{ssec:sim} numerical simulations. At last, section \ref{sec:con} concludes this paper.
\fi
\section{Resilience and Mixed Criticality Concept}\label{sec:resmixcrit}
In this section, we introduce the concepts of resilience and mixed criticality and consequently combine these considerations into a joint metric based on the allocated and desired data rate.\vspace{-.3cm}
\subsection{Resilience}
    Resilience describes the ability of a network or system to provide and maintain an acceptable service level while facing errors or unexpected events that impact the work-flow of the service \cite{STERBENZ20101245}. 
    Hence, resilience is the ability to recover from erroneous conditions or faulty situations \cite{Najarian2019}. In this context, we clearly differentiate between resilience and robustness, whereas resilience is the more general concept which includes robustness along with other aspects, e.g., survivability, dependability, and many more \cite{STERBENZ20101245}.
    Therefore, a fault, error, and failure chain is established by the authors of \cite{STERBENZ20101245} in the network resilience context. Thereby, a fault is a system flaw that can be present on purpose (constraints) or accidentally (software bug, hardware flaw) and cause an observable error. An error is defined as a deviation between the observed and the desired state. A failure is the deviation of service functionality from a desired/required functionality, resulting from an error. In this work, a wireless communication resource management system is considered, which is faulty by nature due to the  unreliability of wireless channels and strict constraints, e.g., transmit power. Hence, errors may manifest as channels outages or hardware failures at the transmitter and/or receiver, which then corresponds to service failures, such as outages or deviations of the provided and requested QoS, e.g., data rate drops. 
    
    In general, redundancy and diversity are typical enablers of fault tolerance and survivability, respectively. 
    Herein, we aim
    to design a resource management mechanism that is resilient such that major service failures do not occur. Additionally, we consider the performance in terms of QoS fulfillment (QoS metrics are delay, throughput, etc. \cite{STERBENZ20101245}). As faults are inevitable, the herein developed resource management should also be able to mitigate the effects of service failures. 
    
    The work \cite{STERBENZ20101245} proposes four vital strategies for the design and assessment of resilient systems: $1)$ \emph{Defending} against threads to normal operation, which can be done actively and passively, e.g., via redundancy and diversity; $2)$ \emph{detection} of erroneous conditions, e.g., via cyclic redundancy checks; $3)$ \emph{remediation} of the erroneous effects, e.g., via automatic adaption of resource allocation; $4)$ \emph{recovery} to normal operations. Especially, strategies $2)-4)$ are shown in Fig.~\ref{overall_figure} in the resilience controller, which detects errors and performs resource allocation on the lower layers. Strategy $1)$ \emph{defending}, is implicitly considered to be part of the initial resource management solution, i.e., we propose passive defense against errors (redundancy and diversity). 
    
    Further, to measure the resilience of a system, the work \cite{Najarian2019} proposes general resiliency metrics. In this work, such metrics are tailored to the physical layer of wireless communication systems to make them applicable in this context. Especially, the considerations capture the resilience aspects of \emph{anticipation}, \emph{absorption}, \emph{adaption}, and \emph{recovery}. Here, \emph{anticipation} is happening before an adverse event (prefailure), which corresponds to $1)$ \emph{defending}. Note that \emph{anticipation} is also not covered in the considered postfailure-related resilience metric. It is assumed to be done a priori by the network operator using established techniques, i.e., we assume the system to be in an optimized state initially until a failure occurs. For more details to prefailure aspects, we refer to works on reliability, e.g., \cite{sota14} and references therein. Next, we characterize the remaining resilience aspects:
    \begin{itemize}
        \item \emph{Absorption}: Measure of the ability to maintain functionality facing errors, i.e., how well a system absorbs a hazard's impact and restrains the severity, corresponding to $1)$ \emph{defending}.
        \item \emph{Adaption}: Measure of the loss of functionality after performance degradation until recovery, i.e., how well the system utilizes existing resources to mitigate the failure consequences, corresponding to $3)$ \emph{remediation}.
        \item \emph{Recovery}: Measure of the ability to recover to a stable state after experiencing degraded functionality, i.e., how fast the system can return to normal (or stable) operation, corresponding to $4)$ \emph{recovery}.
    \end{itemize}
    Note that this work assumes ideal $2)$ detection, i.e., perfect and immediate knowledge of failure conditions.

    A major complication in physical layer communication resource management comes from wireless channels, which are unreliable due to fading, blockage, the nature of electromagnetic radiation, etc. Such behavior can be tackled by introducing diversity techniques, i.e., time/frequency/spatial-diversity, sub-carrier coding, and multiple antennas, respectively \cite{sota13}. While such techniques generally aim at providing reliability/robustness of the wireless communication, we note that resilience includes more aspects which need to be considered from an overall network perspective. Mostly, resilience is implemented as a redundancy mechanism by retransmitting the data. While retransmission is the only feasible solution to recover from an outage, i.e., packet loss, resiliency mechanisms based solely on retransmission are not able to account for long-term outages due to blockage, hardware impairments, or transmitter outages. Transmit devices could face an infinite loop of retransmission on such link failures which deteriorates any spectral and energy efficiency. Many network-layer works consider the resiliency by re-routing traffic, thereby avoiding the failed links. To depart from these works, we consider the resiliency capabilities of lower layers for the resource management in wireless communication systems.\vspace{-.1cm}

\subsection{Mixed Criticality}\label{ssec:mixedcrit}
    The integration and coexistence of data links/flows, which generally have different criticality (importance) levels, into a common communication system is the major challenge of the mixed criticality concept tailored to the communications domain.\footnote{Mixed criticality research originates from a focus on real-time embedded systems with the aim of integrating components of different criticalities into a common hardware platform \cite{burns2013mixed}.} Generally, there is not only a huge number of safety-critical, mission-critical, and low-critical IoT devices, moreover, these criticality levels coexist within one system \cite{sota20}. The concept is demonstrated in Fig,~\ref{overall_figure}, where safety integrity levels (SILs) specify the target level of function safety, for more details we refer to \cite{sil}.
    
    Common principles for mixed criticality in terms of different tasks are given in \cite{burns2013mixed}. Each data flow is defined by its period, deadline, computation time, and criticality level. In \cite{sota21}, this concept is extended to data flows, which are periodic end-to-end communications between source and destination. In this case, each flow is defined by its period, deadline, criticality level, number of hops, and routing path. 
    These considerations are high-level characteristics, i.e., while the physical layer is utilized solely for data transmission, it does not account for mixed criticality. While such mixed criticality characteristics for higher layers provide a certain degree of resilience, a cross layer resilience strategy design is needed as it is essential for the performance to implement criticality level also in physical layer. Hence, departing from the conventional definition of mixed criticality for higher layers, in this work we propose and design common definitions and concepts for mixed criticality in the physical layer.
    To provide methods for supporting such criticality considerations, we define QoS requirements of data links to represent the criticality level. We herein assume that the criticality levels are provided by higher layer algorithms, thereby the QoS demands are given to the underlying layers, which need to account for them, e.g., see Fig.~\ref{overall_figure}.
    
    
    Mixed criticality is usually implemented via weighting the utilities under optimization, e.g., weighted sum rate maximization. However, this doesn't necessary satisfy the requirements and demands of mixed criticality, the weights need to be carefully chosen to incorporate a mixed criticality factor and be updated in an adaptive fashion.
    Other approaches present in literature are the considerations of specific constraints capturing such system demands, e.g., QoS constraints. Such constraints might render a lot of networks infeasible, especially in case the QoS demands are overall hardly achievable. In general, such weight-based or constraint-based approaches are not well suited to provide mixed criticality from a network perspective.
    
    In this work, utilizing the novel approach proposed in \cite{reifert2022energy},
    we consider the mean squared error (MSE) of QoS deviation, i.e., the gap of allocated and desired rate. Thereby, we are able to achieve a good-effort QoS fulfillment and avoid the infeasibility problems of QoS constraints. Hence, we optimize the resource allocation under a mixed critical network in order to fulfill QoS demands subject to various network constraints.
    This approach captures 
    the possibility of having different criticality levels on the physical communication layer.
    
    In the next subsection, we address in details the performance metric capturing the essentials of mixed criticality and resilience. Especially, we propose a time-dependent criticality-aware resilience metric including aspects of absorption, adaption, and recovery.\vspace{-.2cm}
\subsection{Joint Metric}
    As per \cite{STERBENZ20101245}, the considered QoS metrics can be delay, throughput, packet delivery ratio, etc. A common translation of QoS on the physical layer are the data rates of different network participants (throughput, parts of the delay), e.g., \cite{9145363}. Thereby, the QoS assigned to a network participant is designed to match the desired data rates of the participant (target rates).
	The considered resilience metrics are postfailure-related formulations, we consider the system's \emph{absorption}, \emph{adaption}, and \emph{recovery}, in analogy to \cite{Najarian2019}. However, as opposed to \cite{Najarian2019}, a failure in our work is assumed to occur instantly, which allows us to forgo integrals in the metric. 
	
	\begin{figure}[!t]
		\centering
		\includegraphics[scale=1.1]{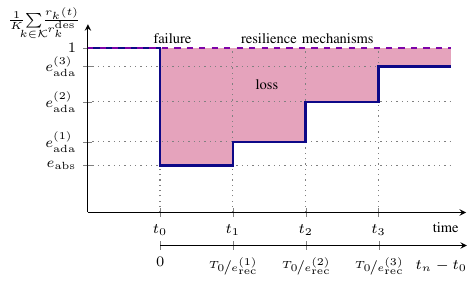}\vspace{-.3cm}
		\caption{Fraction of the allocated and desired sum rate as a function of discrete time points showing the resilience metrics for $n=3$ mechanisms.}
		\label{metric}\vspace*{-.6cm}
	\end{figure}
	
    \paragraph{Absorption} For the \emph{absorption}, we consider the time $t_0$ at which a service failure occurs. Given the time-dependent allocated rate $r_k(t)$ and the constant desired rate $r_k^\mathrm{des}$, the absorption is then defined as the sum of ratios of $r_k(t)$ and $r_k^\mathrm{des}$:
	\begin{equation}\label{eq:r1}
	    e_\mathrm{abs} = \frac{1}{K}\sum_{k\in\mathcal{K}}\frac{ r_k(t_0) }{ r_k^\mathrm{des} },
	\end{equation}
	evaluated at the time of failure $t_0$. 
    As illustrated in Fig.~\ref{metric}, \emph{absorption} can already be observed as performance drop at the time of failure $t_0$. Assuming a system does not enhance its functionality above the demand level, the optimal value of $e_\mathrm{abs}$ is $1$, which would mean the system's performance is unaffected by the failure. As soon as at least one user experiences a lower-than-desired QoS, $e_\mathrm{abs}$ would drop below $1$, resulting in a deviation from the optimum.
    
	\paragraph{Adaption} From the degraded steady state, mechanisms utilizing the existing resources begin to remediate the failure effects. This metric is then defined as ratio of the actual system functionality and the desired functionality at the time $t_n$, where the recovered state is reached, i.e., the time in which an automated remediation mechanism $n$ has finished operating. For an adaption mechanism $n$, this can be formulated as
	\begin{equation}
		e_\mathrm{ada}^{(n)} = \frac{1}{K}\sum_{k\in\mathcal{K}}\frac{ r_k(t_n) }{ r_k^\mathrm{des} }.
	\end{equation}
	If multiple remediation mechanisms are in place, we obtain different values for the adaption which represent the different remediation steps (see Fig.~\ref{metric}). Similar to the absorption, the optimal value of $e_\mathrm{ada}^{(n)}$ is $1$.
	
	\paragraph{Recovery} At last, we consider the time-to-recovery, which is measured between the failure time $t_0$ and the time of recovery $t_n$ and compared to a desired time-to-recovery $T_0$ as
	\begin{equation}
		e_\mathrm{rec}^{(n)} = \begin{cases}
			1 \quad\quad\;\; t_n-t_0 \leq T_0\\
			\frac{T_0}{t_n-t_0} \quad \text{otherwise}
		\end{cases},
	\end{equation}
	where the optimal $e_\mathrm{rec}^{(n)}$ is $1$. For the time-to-recovery, Fig.~\ref{metric} includes a second horizontal axis.
	
	With these considerations at hand, we propose a criticality-aware resilience metric considering the aspects of absorption, adaption, and recovery. The utilized resilience metric is a linear combination of all components defined as 
	\begin{equation}
	    e^{(n)} = \lambda_1 e_\mathrm{abs} + \lambda_2 e_\mathrm{ada}^{(n)} + \lambda_3 e_\mathrm{rec}^{(n)},
	\end{equation}
	where $\lambda_a$ are non-negative weights that satisfy $\sum_{a=1}^3 \lambda_a = 1$. The vector $\boldsymbol{\lambda} = [\lambda_1, \lambda_2, \lambda_3]$ contains the individual weights. 
	
	With the above considerations at hand, a metric for measuring and quantifying the resilience of wireless communications networks is proposed. The challenge remains to, on the one hand include, robust mixed-critical resource management, and on the other hand provide mechanisms for an autonomous controller. Such controller needs to enable smart adaption of resources in the face of system errors, e.g., outages, in order to recover the data rates. In other words, initially, the resources should be allocated to fulfill the mixed-critical QoS demands whilst providing robustness, e.g., by overprovisioning certain resources, or multi-connectivity. In case of system errors, e.g., outage events, resources need to be adapted in order to best restore the performance, e.g., re-allocating some or all resources from the initial solution. The decision on the resource allocation mechanism depends on the desired trade-off between quality and complexity of the adaption.
	
	In what follows, a verification and demonstration of the proposed resilience and mixed criticality framework is provided in terms of a case study. Especially, the case study is related to cloud-based wireless networks, with a focus on RSMA-enabled C-RANs.
	
\section{Case Study: RSMA-enabled C-RAN}\label{sec:casestudy}
    \ifarxiv
    \subsection{Motivation}\label{ssec:moti}
    \fi
    In this section, we investigate the proposed resilience and mixed criticality framework within a network, where various users are connected to multiple base stations (BSs), which are jointly controlled by a central processor (CP) at the cloud, as drawn in \figurename~\ref{overall_figure}. Such C-RAN is a promising network architecture, which enables centralization and virtualization providing high elasticity, high QoS, and good energy efficiency \cite{7378422}. Therefore, it is a suitable candidate to demonstrate the usage of the proposed framework.
    To ensure fulfilling the QoS demands, we investigate QoS target capabilities. Thereby, the QoS assigned to a user is designed to match the desired data rate of the user (target rates). Hence, we aim at designing the C-RAN to enable mixed criticality regarding different communication links. For the communication between the the BSs and the users, we are employing the RSMA as a promising resilience-enhancing paradigm.
    
    
    In this case study, we consider minimizing the MSE of QoS deviation, i.e., the gap of allocated and desired rate. Thereby, we utilize a mixed-critical C-RAN under the RSMA paradigm in order to fulfill all QoS demands. As such, we jointly optimize the precoding vectors and allocated rates subject to per-BS fronthaul capacity, maximum transmit power, and per-user achievable rate constraints as an initial solution. Building upon such system, there are multiple recovery mechanisms, which possibly can be applied and distinguish in the quality of recovery and execution times. In this case study, we apply
    \begin{itemize}
    \item a \emph{rate adaption} mechanism to allocate feasible QoS parameters within the failure condition, 
    \item a \emph{beamformer adaption} mechanism that updates the beamformers according to the new situation,
    \item a \emph{BS-user-allocation adaption} mechanism which optimizes the allocation of BSs to users and
    \item a \emph{comprehensive adaption} mechanism to optimizes the QoS to a value feasible with the selected parameters.
    \end{itemize}
	Upon proposing resilience mechanisms for resource management, we present a resilient RSMA rate management algorithm jointly managing allocated rates, beamforming vectors, and BS-user clustering.\vspace{-.3cm}

    \subsection{System Model}\label{ssec:sysmod}
    The network considered is a downlink C-RAN utilizing \emph{data-sharing} transmission strategy. Under such architecture, a cloud coordinates ${B}$ BSs via fronthaul links in order to serve ${K}$ users, where the CP at the cloud performs most baseband processing tasks. More specifically, the CP encodes messages into signals and designs the joint beamforming vectors. These signals and coefficients are then forwarded to the BSs to perform modulation, precoding, and radio transmission. We denote $\mathcal{B}$ and $\mathcal{K}$ as the set of BSs and single-antenna users, respectively, where the number of BS antennas is $L$. Further network parameters are the fronthaul capacity ${C}_b^{\text{max}}$, the maximum transmit power ${P}_b^{\text{max}}$, and the transmission bandwidth $\tau$. We assume the cloud to have access to the full CSI, which is reasoned in the assumption of a \emph{block-based transmission model}. A transmission block is made of a couple of time slots in which the channel state remains constant, thus the CSI needs to be acquired at the beginning of each block. Without the loss of generality, the proposed algorithm optimizes the resource allocation and provides resiliency within one such block. 
    
    Under the RSMA paradigm, messages requested by users are split into a private and common part. These messages are independently encoded into $s_k^p$ and $s_k^c$, the private and common signal to be transmitted to user $k$. We assume the signals to be zero mean, unit variance complex Gaussian variables with the property of being independent identically distributed and circularly symmetric. Note that under RSMA $s_k^p$ is intended to be decoded by user $k$ only, while $s_k^c$ may be decoded at multiple users, which necessitates a successive decoding strategy.
    
    The channel vector linking user $k$ and BS $b$ is denoted by $\bm{h}_{b,k}\in \mathbb{C}^{L\times 1}$, and we define $\bm{h}_{k}=[(\bm{h}_{1,k})^T,\ldots,(\bm{h}_{B,k})^T]^T\in\mathbb{C}^{L B\times 1}$ as the aggregate channel vector of user $k$. Similar to these definitions, we denote the beamforming vectors as $\bm{w}_{b,k}^o\in\mathbb{C}^{L\times 1}$ and the aggregate beamforming vectors as $\bm{w}_{k}^o=[(\bm{w}_{1,k}^o)^T,\ldots,(\bm{w}_{B,k}^o)^T]^T\in\mathbb{C}^{L B\times 1}$, where $o\in\{p,c\}$ denotes private and common vectors, respectively. Throughout this work, the index $o$ denotes the differentiation of private and common signal related variables.
    
    Due to limited radio resources, BSs naturally have limited capabilities regarding the number of served users. Hence, we introduce the sets $\mathcal{K}_b^p$ and $\mathcal{K}_b^c$, which include only the users whose private or common signal is served by BS $b$. These clusters can be formulated as\vspace{-.1cm}
	\begin{subequations}
	\begin{align}
	    \mathcal{K}_b^p &= \left\{ k \in\mathcal{K} | \text{BS } b \text{ serves } s_k^p \right\},\;\\ \mathcal{K}_b^c &= \left\{ k \in\mathcal{K} | \text{BS } b \text{ serves } s_k^c \right\}.
	\end{align}
	\end{subequations}
	Note that the design of these sets has crucial impact on the system performance. 
	\ifarxiv%
	We provide more details in Appendix~\ref{GAP}.
	\fi%
	Thereby, the previously defined beamforming vectors often contain zeros, i.e., $\bm{w}_{b,k}^o = \bm{0}_L$ when $k\notin\mathcal{K}_b^o$. Each message stream transmits the data via a specific rate $r_k^o$, while the total rate assigned to user $k$ is $r_k = r_k^p + r_k^c$. To ensure operation of the considered network, the CP has to respect the finite fronthaul capacity of the CP-BS links with
	\begin{equation}\label{eq:fthl}
	    \sideset{}{_{k\in\mathcal{K}_b^p}}\sum r_k^p + \sideset{}{_{k\in\mathcal{K}_b^c}}\sum r_k^c \leq C_b^\mathrm{max}.
	\end{equation}
	In what follows, we explain the construction of the transmit signal and the successive decoding scheme for the common streams.
	\subsubsection{Transmit Signal and Successive Decoding}
	Upon receiving the transmit signals and beamforming coefficients, the BSs construct the transmit signal $\bm{x}_b$ by
	\begin{align}
	    \bm{x}_b = \sideset{}{_{k\in\mathcal{K}}}\sum \bm{w}_{b,k}^p s_k^p + \sideset{}{_{k\in\mathcal{K}}}\sum \bm{w}_{b,k}^c s_k^c.\label{eq:x_n_def}
	\end{align}
	Thereby, the signal transmitted by each BS is subject to the maximum transmit power constraint denoted by\vspace{-.1cm}
	\begin{equation}\label{eq:pwr}
	    \mathbb{E}\{ \bm{x}_b^H \bm{x}_b \} = \sum_{{k \in \mathcal{K}}}\Big(\norm[\big]{\bm{w}_{b ,k}^{p}}_2^2 + \norm[\big]{\bm{w}_{b ,k}^{c}}_2^2\Big) \leq P_{b}^{\mathrm{max}}.
	\end{equation}
	
	
	
	In the RSMA framework, multiple users may decode each common message to reduce the interference for messages decoded afterwards. It is thus relevant to consider additional definitions of the network, which are provided as follows:
	
	\begin{itemize}
	    \item The set of users, which decode user $k$'s common message, is defined by $\mathcal{M}_k = \{j\in\mathcal{K}| \text{user } j \text{ decodes } s_k^c\}$.
	    \item The users, whose common messages are decoded by user $k$, are denoted in the set $\mathcal{I}_k = \{j\in\mathcal{K}| k\in \mathcal{M}_j \}$.
	    \item The decoding order at user $k$ is written as $\pi_k$, where $\pi_k(m) > \pi_k(i)$ means that user $k$ decodes common message $i$ before message $m$.
	    \item The set of users, whose common messages are decoded after decoding user $i$'s message at user $k$, become $\mathcal{I}'_{i,k} = \{ m\in\mathcal{I}_k | \pi_k(m) > \pi_k(i) \}$.
	\end{itemize}
	A suitable method of calculating $\mathcal{M}_k$, $\mathcal{I}_k$, $\mathcal{I}'_{i,k}$, and $\pi_k$ is provided by \cite{9445019}. To better illustrate the impact of the above parameters, consider following example: Let user $1$'s common message be decoded by user $1$ and $2$ and user $2$'s common message be decoded only by user $2$, i.e., we obtain $\mathcal{M}_1 = \{1,2\}$ and $\mathcal{M}_2 = \{2\}$. Consequently, we have $\mathcal{I}_1 = \{1\}$ and $\mathcal{I}_2 = \{1,2\}$, meaning that user $1$ decodes its own, and user $2$ decodes both its own and user $1$'s common message. The decoding orders could be $\pi_1 = \{1\}$ and $\pi_2 = \{2>1\}$, where the latter describes user $2$ decoding common message $1$ before its own message. Hence, $\mathcal{I}'_{1,1} = \varnothing$, $\mathcal{I}'_{2,2} = \varnothing$, and $\mathcal{I}'_{1,2} = \{2\}$.

	Based on the previous definitions, we obtain the received signal at user $k$ as
	\begin{align}
	    &y_k = \bm{h}_{k}^H\bm{w}_{k}^p {s}_{k}^p + \sum_{j \in \mathcal{I}_{k}}\bm{h}_{k}^H\bm{w}_{j}^c {s}_{j}^c + \sum_{m \in \mathcal{K}\setminus \{k\}}\bm{h}_{k}^H\bm{w}_{m}^p {s}_{m}^p \nonumber\\
	    &\qquad + \sum_{l \in \mathcal{K}\setminus \mathcal{I}_k}\bm{h}_{k}^H\bm{w}_{l}^c {s}_{l}^c+ n_k .\label{eq:y_k}
	\end{align}
	Here, $n_k \sim \mathcal{C}\mathcal{B}(0,\sigma^2)$ represents additive white Gaussian noise, assumed to have the same power for all users. In \eqref{eq:y_k}, the first two terms consist of private and common signals, which are decoded during the successive decoding. In contrast, the last three terms in \eqref{eq:y_k} denote interference from private signals, common signals, and noise, respectively. Using this definition, the signal to interference plus noise ratio (SINR) of user $k$ decoding its private message, i.e., $\Gamma_{k}^p$, and the common message of user $i$, i.e., $\Gamma_{i,k}^c$, are formulated respectively as
	\begin{subequations}
    \begin{align}
		\label{eq:e2.18}
		\Gamma_{k}^p &= \frac{\left|\bm{h}_{k}^{H}\bm{w}_{k}^p \right|^2}{\sum\limits_{j \in \mathcal{K}\setminus \{k\}}\left|\bm{h}_{k}^{H}\bm{w}_{j}^p \right|^2 + \sum\limits_{l \in \mathcal{K}\setminus \mathcal{I}_k}\left|\bm{h}_{k}^{H}\bm{w}_{l}^c \right|^2+\sigma^2},\\
		\label{eq:e2.19}
		\Gamma_{i,k}^c &=\frac{\left|\bm{h}_{k}^{H}\bm{w}_{i}^c \right|^2}{\sum\limits_{j \in \mathcal{K}}\left|\bm{h}_{k}^{H}\bm{w}_{j}^p \right|^2 + \sum\limits_{l \in \mathcal{K}\setminus \mathcal{I}_k\cup \mathcal{I}'_{k,i}}\left|\bm{h}_{k}^{H}\bm{w}_{l}^c \right|^2+\sigma^2 }.
	\end{align}
	\end{subequations}
	\subsubsection{Mixed Criticality-aware MSE Metric}
	To analyze the MSE of QoS deviation, i.e., the gap of desired and allocated rate, we optimize the metric $\Psi$. In this context, the allocated rates are dependent on different system states with different resource allocation, e.g., initial state, post-failure states. 
	Such metric is formulated mathematically as
	\begin{align}\label{eq:f}
	    &\Psi = \frac{1}{|\mathcal{K}|}\sum_{k\in\mathcal{K}}\left|\left(r_k^p + r_k^c\right) - r^{\mathrm{des}}_k\right|^2.
	\end{align}
	Note that \eqref{eq:f} takes different values at different times, as the resource allocation is updated in a resilient manner, e.g., see Fig.~\ref{metric}. 
	In \eqref{eq:f}, the mixed criticality of links is represented in $r^{\mathrm{des}}_k$, where critical applications have greater QoS requirements than others. 
	\begin{remark}
	Minimizing a function based on the MSE such as \eqref{eq:f} (alternatively based on the mean absolute error) contributes to finding a good trade-off between minimizing the resource usage and providing QoS. On the one hand, systems which can hardly fulfill the desired QoS will fall into a mode of maximizing each user's rate upon meeting those demands. On the other hand, as rate targets are met, no further enhancement of the rates is conducted, so as to avoid wasting resources which provides a form of energy efficiency.
	\end{remark}
	Note that an MSE-based metric as \eqref{eq:f} does not guarantee QoS fulfillment, especially in networks with insufficient resources. However, such objective results in a solution, which best approaches the (unreachable) target rates.
	
	\subsection{Problem Formulation and Convexification}\label{ssec:opt}
	This paper considers the problem of minimizing the constrained network-wide rate gap (QoS deviation) as initial problem, which can be formulated mathematically as follows:%
	\begin{subequations}\label{eq:Opt1}%
		\begin{align}
			\underset{\bm{w},\bm{r},\bm{\mathcal{K}}}{\mathrm{min}}\quad &\Psi  \tag{\ref{eq:Opt1}} \\
			\mathrm{s.t.} \quad &\eqref{eq:fthl}, \eqref{eq:pwr}, \nonumber\\
			& r_{k}^{p} \leq \tau\,\log_2(1+\Gamma_{k}^p),  &\forall k &\in \mathcal{K}, \label{eq:achp}\\	
			& r_{i}^{c} \leq \tau\,\log_2(1+\Gamma_{i,k}^c), &\forall i\in\mathcal{I}_k, \forall k &\in \mathcal{K}. \label{eq:achc}
		\end{align}
	\end{subequations} 
	The above problem minimizes the gap of assigned (private and common) rate to the desired rate by jointly managing the allocated rates, precoding vectors, and BS-user clustering. Hereby, the precoding and rate vector\vspace{-.2cm}
	\begingroup\addtolength{\jot}{-.1cm}
	\begin{align*}
	    \bm{w}&=\Big[\big(\bm{w}_{1}^p\big)^T,\ldots,\big(\bm{w}_{K}^p\big)^T,\big(\bm{w}_{1}^c\big)^T,\ldots,\big(\bm{w}_{K}^c\big)^T\Big]^T, \\
	    \bm{r}&=\big[r_{1}^p,\ldots,r_{K}^p,r_{1}^c,\ldots,r_{K}^c\big]^T,
	\end{align*}\endgroup
	and the clustering set $\bm{\mathcal{K}} = \{ \mathcal{K}_b^o | \forall b\in\mathcal{B}, \forall o\in\{p,c\} \}$	denote the optimization variables. The feasible set of problem \eqref{eq:Opt1} is defined by the fronthaul capacity constraints per BS \eqref{eq:fthl}, the maximum transmit power constraint per BS \eqref{eq:pwr}, the achievable rates per user \eqref{eq:achp} and \eqref{eq:achc}. The latter constraints, namely \eqref{eq:achp} and \eqref{eq:achc}, depend on the SINR defined in \eqref{eq:e2.18} and \eqref{eq:e2.19}, which is non-convex. In contrast, the objective function \eqref{eq:Opt1} is already in convex form, since the squared absolute value is convex. However, as the constraints \eqref{eq:achp} and \eqref{eq:achc} are non-convex, problem \eqref{eq:Opt1} is in general non-convex and difficult to solve directly. Therefore, this paper proposes various reformulation techniques to devise the optimization problem in a more tractable form.
	
	A few notes on problem \eqref{eq:Opt1}'s constraints. Constraint \eqref{eq:e2.19} takes the form of a multiple access constraint. Specifically, $r_i^c$ is the rate of user $i$'s common stream and the subset of users $\mathcal{M}_i$ are going to decode that specific signal. Therefore, $r_i^c$ is bounded by the lowest $\Gamma_{i,k}^c$, i.e., the lowest SINR of all decoding users. Also, the constraints \eqref{eq:fthl} and \eqref{eq:pwr} are highly dependent on the clustering sets $\mathcal{K}_b^p$ and $\mathcal{K}_b^c$, which are also subject to optimization. Hence, problem \eqref{eq:Opt1} has non-convex constraints and is thus non-convex. Therefore, we propose a convex reformulation of problem \eqref{eq:Opt1} in Theorem~\ref{th:1.convex} based on fractional programming, inner-convex approximation, and a $l_0$-norm approximation clustering approach. Before stating Theorem~\ref{th:1.convex}, consider the following definitions: Let $\bm{u}$ be an auxiliary variables vector defined as $\bm{u}=[(\bm{u}_{1}^p)^T,\ldots,(\bm{u}_{K}^p)^T,(\bm{u}_{1}^c)^T,\ldots,(\bm{u}_{K}^c)^T]^T$, consisting of $\bm{u}_k^o=[u_{1,k}^o,\ldots,u_{B,k}^o]^T$. Similarly, the auxiliary variable $\bm{\gamma}$ covers all possible non-zero elements of $\bm{\gamma}'=[{\gamma}_{1}^p,\ldots,{\gamma}_{K}^p,{\gamma}_{1,1}^c,{\gamma}_{1,2}^c,\ldots,{\gamma}_{K,K}^c]^T$, 
	as each common message is only decoded by a part of all users. 
	Let $\beta_{b,k}^o$ denote weights of the $l_0$-norm approximation. The functions $g^p(\bm{w},\bm{\gamma})$ and $g^c(\bm{w},\bm{\gamma})$ are defined as
	\begin{align}
		&g^p(\bm{w},\bm{\gamma}) = \gamma_{k}^p - 2 \,\mathrm{Re}\Big\{({\chi}_{k}^p)^* \left|\bm{h}_{k}^{H}\bm{w}_{k}^p \right|^2\Big\} + |\chi_{k}^p|^2 \nonumber\\
		&\qquad \cdot\left[ \sigma^2 + \sum\limits_{j \in \mathcal{K}\setminus \{k\}}\left|\bm{h}_{k}^{H}\bm{w}_{j}^p \right|^2 + \sum\limits_{l \in \mathcal{K}\setminus \mathcal{I}_k}\left|\bm{h}_{k}^{H}\bm{w}_{l}^c \right|^2 \right],\hspace{-0.1cm} \label{eq:qtp} \\
		&g^c(\bm{w},\bm{\gamma}) = \gamma_{i,k}^c - 2 \,\mathrm{Re}\Big\{(\chi_{i,k}^c)^* \left|\bm{h}_{k}^{H}\bm{w}_{i}^c \right|^2\Big\} + |\chi_{i,k}^c|^2\nonumber\\
		&\qquad\cdot \left[ \sigma^2 + \sum\limits_{j \in \mathcal{K}}\left|\bm{h}_{k}^{H}\bm{w}_{j}^p \right|^2 + \sum\limits_{l \in \mathcal{K}\setminus \mathcal{I}_k\cup \mathcal{I}'_{k,i}}\left|\bm{h}_{k}^{H}\bm{w}_{l}^c \right|^2 \right],\hspace{-0.1cm} \label{eq:qtc}
	\end{align}
	where ${\chi}_{k}^o$ denote auxiliary variables of the quadratic transform. At last, the reformulated fronthaul constraint \eqref{eq:fthl} is given as\vspace{-.1cm}
	\begingroup\addtolength{\jot}{-.12cm}
	\begin{align}
	    &\hspace{-.3cm}\sum_{{k \in \mathcal{K}}}\hspace{-.1cm} \Big(\big(u_{b,k}^p+r_{k}^{p}\big)^2-2 \big(\tilde{u}_{b,k}^p-\tilde{r}_{k}^{p}\big)\big(u_{b,k}^p-r_{k}^{p}\big)+\big(\tilde{u}_{b,k}^p-\tilde{r}_{k}^{p}\big)^2 \nonumber\hspace{-2.5cm}\\
		&\qquad+ \big(u_{b,k}^c+r_{k}^{c}\big)^2-2 \big(\tilde{u}_{b,k}^c-\tilde{r}_{k}^{c}\big)\big(u_{b,k}^c-r_{k}^{c}\big)+\nonumber\hspace{-2.5cm}\phantom{\Big)}\\ &\qquad+\big(\tilde{u}_{b,k}^c-\tilde{r}_{k}^{c}\big)^2\Big) \leq 4C_n^{\mathrm{max}},\hspace{-2.5cm}  &\forall b &\in \mathcal{B}, \label{eq:fthl2}
	\end{align}\endgroup
    where $\tilde{\bm{u}}$ and $\tilde{\bm{r}}$ are feasible fixed values, which originate from the inner-convex approximation.
	\begin{theorem}\label{th:1.convex}
	    Problem \eqref{eq:Opt1} can be transformed into the following convex reformulation based on fractional programming, inner-convex approximation, and $l_0$-norm approximation:
	    \begin{subequations}\label{eq:Opt2}
		\begin{align}
			\underset{\bm{w},\bm{r},\bm{u},\bm{\gamma}}{\mathrm{min}}\quad &\Psi \tag{\ref{eq:Opt2}} \\
			\mathrm{s.t.} \quad &\eqref{eq:pwr}, \eqref{eq:fthl2}, \nonumber\\
			& r_{k}^{p} \leq \tau\,\log_2(1+\gamma_{k}^p),  &\forall k &\in \mathcal{K}, \label{eq:achp2}\\	
			& r_{i}^{c} \leq \tau\,\log_2(1+\gamma_{i,k}^c), &\forall i\in\mathcal{I}_k, \forall k &\in \mathcal{K},\label{eq:achc2}\\
			& g^p(\bm{w},\bm{\gamma}) \leq 0,  &\forall k &\in \mathcal{K}, \label{eq:gpprob}\\	
			& g^c(\bm{w},\bm{\gamma}) \leq 0, &\forall i\in\mathcal{I}_k, \forall k &\in \mathcal{K},\label{eq:gcprob}\\
			& \beta_{b,k}^p \, \norm[\big]{\bm{w}_{b ,k}^{p}}_2^2 \leq u_{b,k}^p,\; \beta_{b,k}^c \, \norm[\big]{\bm{w}_{b ,k}^{c}}_2^2 \leq u_{b,k}^c,\nonumber\hspace{-2.65cm}\\
			& &\forall b\in\mathcal{B}, \forall k &\in \mathcal{K}.\label{eq:betacon}
		\end{align}
	\end{subequations}
	\end{theorem}
	\begin{proof}
	    For a detailed derivation, we refer to Appendix \ref{app3}.
	\end{proof}
	\ifarxiv%
	An efficient procedure to solve the resource allocation problem of the considered network, i.e., problem \eqref{eq:Opt2}, can be found in Appendix~\ref{appB}.
	\fi%
	With the above considerations at hand, problem \eqref{eq:Opt2} is solved as an initial resource management step. The objective function defined in \eqref{eq:f} provides the necessary means for mixed criticality-awareness of the resource allocation. Additionally, in networks with sufficient resources, e.g., fronthaul capacity, problem \eqref{eq:Opt2} yields a robust solution, providing high levels of \emph{absorption}. This is due to the fact that the rates are allocated to meet the QoS, while the actual achievable rates might be much higher, i.e., so called SINR margins. Also, network participants are multi-connected to multiple BSs. In resource limited networks, however, i.e., fronthaul capacity limited regimes, a solution based on problem \eqref{eq:Opt2} yields worse robustness, e.g., due to single-link connections, but achieves a best-effort QoS fulfillment solution. In other words, QoS targets might not be fully satisfied, whereas, due to the nature of $\Psi$ \eqref{eq:f}, the priority lays in meeting high critical targets.
	\subsection{Resilient and Criticality-aware Resource Allocation}\label{ssec:resmixcritalloc}
    In this work, we utilize four different resilience mechanisms, which differ in calculation time and quality of recovery. Hence, the four algorithms can run in parallel to sequentially recover from the outage. 
	\paragraph{Resiliency Mechanism 1 (M1)} As first mechanism in line, \emph{rate adaption} is used.
	Hereby, the algorithm calculates the achievable rates of all users, based on the new SINR measured after occurrence of the failure, while the beamformers are kept fixed. 
	Thereby, from comparing allocated and achievable rate, the algorithm determines which users experience outage. The rate of these users is then set to the new achievable rate, which will be lower than the previously allocated rate and, in fact, might be zero. By doing so, M1 recovers the communication links to users, which otherwise could not decode their messages after outage due to unadjusted rates. 
	\emph{Rate adaption} is the resilience mechanism of the lowest complexity ($\mathcal{O}(1)$) and thus, also the fastest. However, it is also the weakest recovery mechanism, i.e., most recovered rates will only be a fraction of the users desired rate. We associate $e_\mathrm{ada}^{(1)}$ and $e_\mathrm{rec}^{(1)}$, i.e., the adaption and recovery metrics, with this mechanism.
	
	
	\paragraph{Resiliency Mechanism 2 (M2)} As second resilience mechanism in line, \emph{beamformer adaption} is employed.
	This scheme solves a reduced-complexity version of problem \eqref{eq:Opt2} and thereby calculates new beamforming coefficients while keeping the clustering fixed. The utilization of a previous feasible solution to initialize the procedure reduces the time to convergence. 
	This mechanism has an intermediate complexity ($d_2={K}(2{B}{L}+{K}+3)$) but, compared to M1, leads to a higher recovered QoS. In comparison, this mechanism manages the interference more efficiently utilizing spatial dimensions. We associate $e_\mathrm{ada}^{(2)}$ and $e_\mathrm{rec}^{(2)}$.
	
	\paragraph{Resiliency Mechanism 3 (M3)} This mechanism is referred to as \emph{BS-user-allocation adaption} and repeats the network's clustering using the updated CSI by solving a generalized assignment problem. For more details on solving the generalized assignment problem%
	\ifarxiv%
	, we refer to Appendix~\ref{app4}.
	\fi%
	Afterwards, the remaining mechanism operates analog to M2. This mechanism has a high complexity ($d_{3,1}={K}{B}$ and $d_{3,2}={K}(2{B}{L}+{K}+3)$) but offers good-quality recovered QoS, we associate $e_\mathrm{ada}^{(3)}$ and $e_\mathrm{rec}^{(3)}$.
	
	\paragraph{Resiliency Mechanism 4 (M4)} At last, we employ \emph{comprehensive adaption}, which repeats the solution to problem \eqref{eq:Opt2}, solving the network's clustering, beamforming, power control, and rate allocation jointly using the updated CSI. This mechanism has a very high complexity ($d_4={K}(2{B}({L}+1)+{K}+3)$) and thus it is rather optimistic to use this technique in practice. However, this adaption can either be seen as an mechanism for networks with slowly changing channels or as bound and approximate for other resilience strategies. This yields $e_\mathrm{ada}^{(4)}$ and $e_\mathrm{rec}^{(4)}$.
	
	\begin{figure*}[t]
	\centering
	\begin{subfigure}[t]{0.32\textwidth}
		\centering
	    \includegraphics[width=1.05\linewidth]{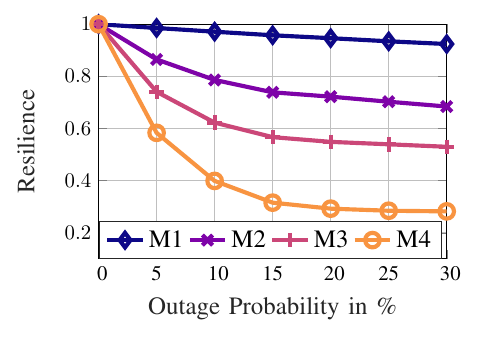}\vspace{-.4cm}
	    \caption{Time-to-recovery under focus, $\boldsymbol{\lambda} = [0.2,\;0.0,\;0.8]$.}
	    \label{fig:yRes_xOut_allM_lambda_02_00_08}
	\end{subfigure}\hfill
	\begin{subfigure}[t]{0.32\textwidth}
		\centering
	    \hspace{.2cm}\includegraphics[width=.95\linewidth]{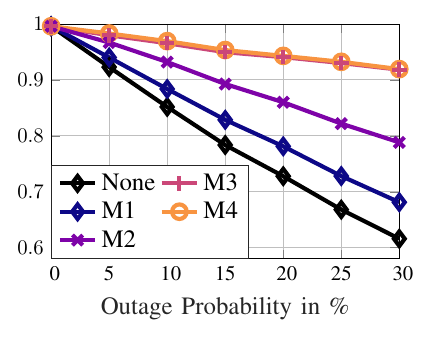}\vspace{-.4cm}
	    \caption{Quality-of-adaption under focus, $\boldsymbol{\lambda} = [0.2,\;0.8,\;0.0]$.}
	    \label{fig:yRes_xOut_allM_lambda_02_08_00}
	\end{subfigure}\hfill
	\begin{subfigure}[t]{0.32\textwidth}
		\centering
	    \includegraphics[width=.95\linewidth]{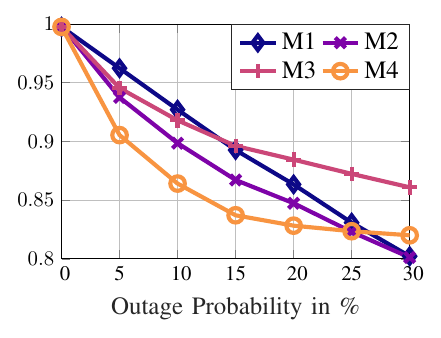}\vspace{-.4cm}
	    \caption{Trade-off adaption and recovery, $\boldsymbol{\lambda} = [0.2,\;0.4,\;0.4]$.}
	    \label{fig:yRes_xOut_allM_lambda_02_06_02}
	\end{subfigure}\vspace{-.2cm}
	\caption{Resilience over outage probability comparing the mechanisms.} \label{yRes_xOut_allM}\vspace{-.6cm}
    \end{figure*}
    \begin{figure*}[t]
	\centering
	\begin{subfigure}[t]{0.48\textwidth}
		\centering
	    \includegraphics[width=.9\linewidth]{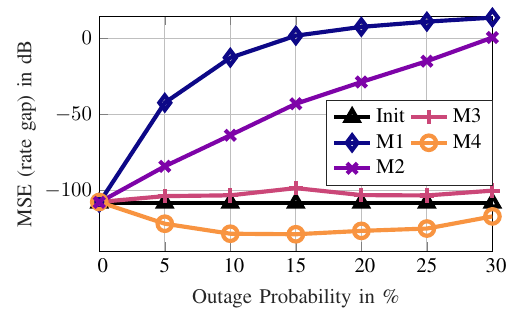}\vspace{-.4cm}
	    \caption{MSE in dB comparing initial solution and all mechanisms.}
	    \label{fig:yMSE_xOut_allM}
	\end{subfigure}
	\hfill
	\begin{subfigure}[t]{0.48\textwidth}
	\centering
	    \includegraphics[width=.9\linewidth]{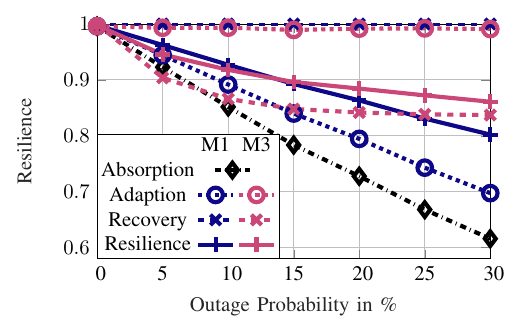}\vspace{-.4cm}
	    \caption{Resilience components of M1 and M2, with $\boldsymbol{\lambda} = [0.2,\;0.4,\;0.4]$.}
	    \label{fig:yRes_xOut_M13_lambda_02_06_02v2}
	\end{subfigure}\vspace{-.2cm}
	\caption{Resilience and MSE over outage probability comparing various mechanisms.} \label{yResMSE_xOut}\vspace{-.6cm}
    \end{figure*}%
    
	The full mechanism, which detects outages and then applies the resilience mechanisms is referred to as \emph{Resilient and Criticality-aware RSMA Resource Management} algorithm and can be found in Algorithm~\ref{AlgCen}. The algorithm gets as input the desired rates of all user, as these rates are typically specified by the users. 
	First, the initial solution resource allocation strategy is computed by solving problem \eqref{eq:Opt2}. 
	For each transmission, the algorithm checks if an outage occurs by checking if the new achievable rate of any user is smaller than the allocated rate. Note that $\varepsilon$ refers to a small tolerance threshold value to increase robustness. If an outage or failure happens, 
	the algorithm sequentially executes all four resilience mechanism, namely \emph{rate adaption},  
	\emph{beamformer adaption}, \emph{BS-user-allocation adaption}, and \emph{comprehensive adaption}, as described previously. Note that the algorithm stops as soon as any mechanism achieves the desired functionality, otherwise it finished after \emph{comprehensive adaption} ends. Thereby, we assume that no second failure occurs while executing the recovery mechanisms. Afterwards, the algorithm continues checking each transmission for outages.

    \begin{algorithm}[h]
    \caption{Resilient and Criticality-aware RSMA Resource Management}
    \begin{algorithmic}[1]
    \REQUIRE desired rate $r_k^\mathrm{des}$ $\forall k\in\mathcal{K}$
	\STATE $\bm{r}$, $\bm{w}$, $\bm{u}$, $\bm{\gamma}$ $\gets$ solution from \eqref{eq:Opt2}
	\WHILE{true}
	\STATE $t_0$, $\bm{r}(t_0)$ $\gets$ latest transmission time and achievable rate
	\STATE\textit{$\triangleright$\ Obtain channel outage information}
	\IF{$\exists k: r_k(t_0) < r_k-\varepsilon$}
	\STATE\textit{$\triangleright$\ Resiliency Mechanism 1 (rate adaption)}
	\STATE $\bm{r}\gets$ new achievable rates of \eqref{eq:achp} and \eqref{eq:achc}
	\STATE\textit{$\triangleright$\ Resiliency Mechanism 2 (beamformer adaption)}\vspace*{-.05cm}
	\STATE $\bm{r}$, $\bm{w}$ $\gets$ new solution to \eqref{eq:Opt1}, starting from the previously feasible solution\\ \vspace*{-.05cm}
	\STATE\textit{$\triangleright$\ Resiliency Mechanism 3 (BS-user-allocation adaption)}\vspace*{-.05cm}
	\STATE $\mathcal{K}_b^p$, $\mathcal{K}_b^c$ $\forall b\in\mathcal{B}$ $\gets$ new solution of clustering in  \eqref{eq:GAP}
	\STATE $\bm{r}$, $\bm{w}$ $\gets$ new solution to \eqref{eq:Opt1} \vspace*{-.05cm}
	\STATE\textit{$\triangleright$\ Resiliency Mechanism 4 (comprehensive adaption)}\vspace*{-.05cm}
	\STATE $\bm{r}$, $\bm{w}$, $\bm{u}$, $\bm{\gamma}$ $\gets$ new solution to \eqref{eq:Opt2} \vspace*{-.05cm}
	\ENDIF
	\ENDWHILE
    \end{algorithmic}
    \label{AlgCen}
    \end{algorithm}%
    
	\begin{remark}
	    When having the choice between multiple resilience mechanism, the resilience mechanisms with the fastest computation time should be scheduled first, while the ones leading to the highest recovery quality should be scheduled last. Thereby, it depends on the different computation times and recovery qualities, if it is optimal to use all mechanisms, or just a few of them.
	\end{remark}\vspace{-.4cm}
	
    \subsection{Numerical Simulation}\label{ssec:sim}
	To evaluate the performance of the proposed methods, we conduct numerical simulations in this section. We consider a network over a square area of $\SI{800}{\meter}$ by $\SI{800}{\meter}$, in which BSs and users are placed randomly. Each BS is equipped with $L=2$ antennas and has a maximum transmit power of $P_b^\mathrm{max} = \SI{28}{dBm}$. Unless mentioned otherwise, we fix $C_b^\mathrm{max}=\SI{50}{Mbps}$, the number of BS $B=6$, and the number of users $K=14$. Regarding M3, the assignment problem related parameters are set as $B_k^\mathrm{max} = 2$ and $K_n^\mathrm{max} = 2\cdot K$. 
	We consider a channel bandwidth of $\tau = \SI{10}{MHz}$, and a path-loss model given by $\mathrm{PL}_{b,k}= 128.1 + 37.6\cdot\mathrm{log}_{10}(d_{b,k})$, where the distance of user $k$ and BS $b$ is denoted as $d_{b,k}$. Additionally, we consider log-normal shadowing with $8$ dB standard deviation and Rayleigh fading with zero mean and unit variance. The noise power spectral density is $\SI{-168}{dBm/Hz}$. Unless mentioned otherwise, we set the mixed-critical QoS demands $r_k^\mathrm{des}$ to $\SI{12}{Mbps}$, $\SI{8}{Mbps}$, and $\SI{4}{Mbps}$, for $4$ random users each, respectively. This corresponds to three criticality levels, namely high (HI), medium (ME), and low (LO).
	
	\begin{figure*}[t]
	\centering
	\begin{subfigure}[t]{0.48\textwidth}
	    \centering
	    \includegraphics[width=.9\linewidth]{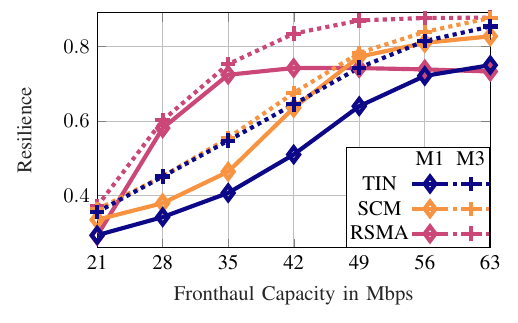}\vspace{-.4cm}
	    \caption{Resilience of TIN, SMC, and RSMA for M1 and M3.}
	    \label{fig:yRes_xFthl_M14_lambda_03_065_005}
	\end{subfigure}\hfill
	\begin{subfigure}[t]{0.48\textwidth}
	    \centering
	    \includegraphics[width=.9\linewidth]{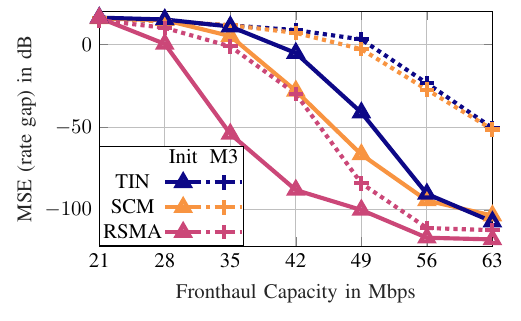}\vspace{-.4cm}
	    \caption{MSE in dB comparing initial solution and M3.}
	    \label{fig:yMSE_xFthl_MI4}
	\end{subfigure}\vspace{-.2cm}
	\caption{Resilience and MSE over fronthaul capacity comparing interference management schemes.} \label{yResMSE_xFthl}\vspace{-.6cm}
    \end{figure*}%
	\subsubsection{Resilience Components}
	At first, we aim at comparing the proposed mechanisms in term of resilience and analyzing the individual components of resilience, i.e., absorption, adaption, recovery, and how the combine.\\
	\indent
	In the first set of simulations, we compare the performance of the four considered schemes in terms of resilience, i.e., $e^{(n)}$ as per \eqref{eq:f}, where $n\in\{1,\cdots,4\}$. Fig.~\ref{yRes_xOut_allM} shows three plots from the same data set, differing in the exact weightings $\boldsymbol{\lambda}$ in \eqref{eq:f}. That is, Fig.~\ref{fig:yRes_xOut_allM_lambda_02_00_08} shows the resilience when the \emph{recovery} (time-to-recovery) is weighted most, i.e., systems where timely recovery is valued over all other aspects, over the channel outage probability. As expected, M1 and M2 provide the fastest recovery due to inhabiting the lowest complexity, however, this result does not capture the quality-of-adaption. The resilience of M3 and M4 decreases with increasing outage probability. Interestingly, we observe a stabilization at approximately $15$\%, implying that this parameter has only slight impact on the calculation time of M3 and M4. In contrast, M2 and especially M1 are steadily decreasing with rising outage probability. This behavior comes due to the fact that M1's computation time is strongly coupled to the amount of users experiencing outages. M2 is in part dependent on the quality of the previously feasible solution, and with increasing outages, the previous solution looses its value as an initial point.\\
	In Fig.~\ref{fig:yRes_xOut_allM_lambda_02_08_00}, the resilience for setups, which value \emph{adaption} most, is plotted over outage probability. Results emphasize the quality-of-adaption difference of the mechanisms, as M1 performs worst, M2 medium, and M3 as well as M4 best. 
	Interestingly, we observe M3 closely approaching M4, which, for this network setup, highlights M3's superiority, as it is less complex than M4 providing almost the same \emph{adaption}. Overall, Fig.~\ref{fig:yRes_xOut_allM_lambda_02_08_00} underlines the need for resilience, especially in networks where channel outages occur, i.e., all realistic setups, as all mechanisms outperform the \emph{no-reaction} baseline (None), where no resilience action is taken.\\
	A combination of both results is shown in Fig.~\ref{fig:yRes_xOut_allM_lambda_02_06_02}. We show the resilience versus the outage probability for a reasonable trade-off weight vector. It comes clear that a dynamic switching between M1 and M3 provides the best resilience over the considered outage probabilities. While at the lower probabilities, such scheme mostly benefits from M1's quick recovery, at higher values, M3's quality-of-adaption are most beneficial.\\
    \indent    
    From Fig.~\ref{fig:yMSE_xOut_allM}, we get a different point of view, as the figure shows the objective function \eqref{eq:Opt1}, i.e., the MSE in dB as a function of outage probability. The initial performance (Init) depicts the performance baseline of the algorithm for $0$\% outage probability. While M1 and M2 are significantly affected by the rising outages, M3 performs close to and M4 even outperforms Init. This emphasizes the fact that a re-clustering approach as M3 can almost keep the initial performance level facing outages. The performance gap between M3 and M4 can be seen as a dynamic clustering benefit, as such approach can lower the MSE even further. Note that outages do not solely impact active channels, channels that cause interference (e.g., unassociated BS and user) are also subject to outage. Thus, M4 can beat the initial performance facing outages.\\
	The individual aspects of the proposed resilience metric are shown in Fig.~\ref{fig:yRes_xOut_M13_lambda_02_06_02v2} for M1 and M3. Here, absorption, adaption, recovery, and the resulting total resilience are plotted. Consistent with previous observations, M1 has a better time-to-recovery, yet worse adaption and thus worse total resilience, than M3. Note that the herein considered results are highly dependent on the weight vector $\boldsymbol{\lambda}$, which has to be tailored to the needs of different services, networks, and providers.
	
	\ifarxiv
	\begin{figure*}[t]
	\centering
	\begin{subfigure}[t]{0.48\textwidth}
	    \centering
	    \includegraphics[width=.9\linewidth]{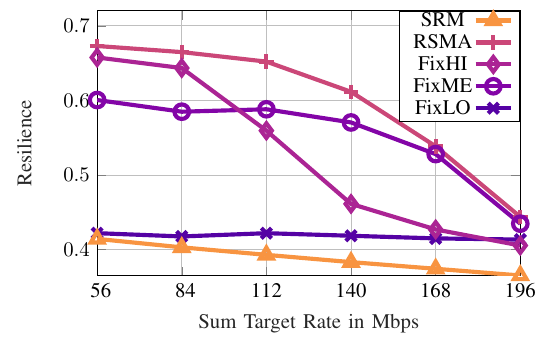}\vspace{-.4cm}
	    \caption{Resilience of M4, $\boldsymbol{\lambda} = [0.2,\;0.5,\;0.3]$.}
	    \label{fig:yRes_xSumQOS_M5_lambda_02_05_03}
	\end{subfigure}\hfill
	\begin{subfigure}[t]{0.48\textwidth}
	    \centering
	    \includegraphics[width=.9\linewidth]{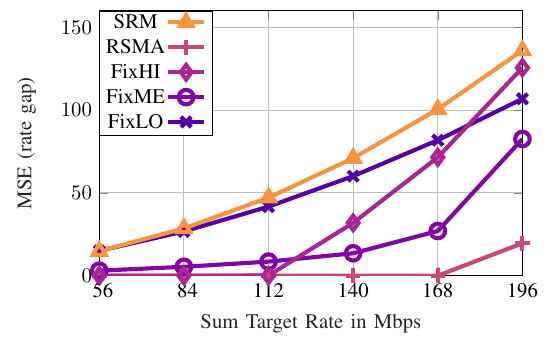}\vspace{-.4cm}
	    \caption{MSE performance for M4.}
	    \label{fig:yMSE_xSumQOS_MI5}
	\end{subfigure}\vspace{-.2cm}
	\caption{Resilience and MSE over QoS requirements comparing criticality schemes.} \label{yResMSE_xSumQOS}\vspace{-.6cm}
    \end{figure*}
    \fi
	\subsubsection{Impact of Interference Management Techniques}
	In addition to the herein proposed scheme, referred to as RSMA, we consider $2$ different reference schemes for interference management, namely treating interference as noise (TIN) and a single common message (SCM)-based RSMA scheme. TIN does not consider any rate-splitting capabilities, i.e., TIN is less complex than RSMA, yet offers less opportunities for resilience, e.g., redundant and diverse message streams. Contrary to that, the SCM scheme employs one-layer rate-splitting, where one \emph{super common} stream (additional message stream) is decoded by all users in addition to private messages, e.g., used in \cite{8846706,7513415}. SCM requires each user to decode two messages, therefore, has higher complexity than TIN, but lower complexity than RSMA. Hence, such complexity differences have impact on the time-to-recovery $e_\mathrm{rec}^{(n)}$, as computing times differ. As a highlight of the proposed resilience metric, it captures these computation time differences numerically within the overall resilience metric \eqref{eq:f}. Thereby, the comparison of RSMA, SCM, and TIN in terms of resilience is fair by nature, as it includes complexity differences.\\
    \indent	
	Fig.~\ref{fig:yRes_xFthl_M14_lambda_03_065_005} shows the resilience of M$1$ and M$3$ using TIN, SCM, and RSMA as a function of the fronthaul capacity for an outage probability of $25$\%, with $\boldsymbol{\lambda} = [0.3,\;0.65,\;0.05]$. Observing M1, all three plots have a similar starting point at ${C}_b^{\text{max}} = \SI{21}{Mbps}$ and saturate at higher capacities. Especially in the medium fronthaul regime, i.e., mostly relevant in practice, (here, ${28-49}$~Mbps), RSMA provides significant resilience enhancements compared to TIN and SCM. Similarely, SCM is able to outperform TIN in every point of the x-axis. This result highlights the potential gain from using rate-splitting-based schemes for ensuring resilience of communication networks. In the higher fronhaul regime, where QoS targets are easily achievable, RSMA is outperformed by SCM (and TIN), as time-to-recovery gains importance over quality-of-adaption, especially for M1, which is the fastest mechanism. Considering M3, we observe that RSMA outperforms SCM and TIN for all fronthaul capacities, emphasizing the gain of the proposed scheme in terms of resilience.\\
	To gain insights into the algorithm's behavior in terms of rate gap, i.e., MSE objective \eqref{eq:f}, Fig.~\ref{fig:yMSE_xFthl_MI4} plots $\Psi$ in dB as a function of ${C}_b^{\text{max}}$ for the initial solution (Init), i.e., performance before outage, and M3. RSMA's initial performance is able to provide the lowest MSE among all interference management schemes, saturating at high fronthaul capacity, as QoS targets are already met up to a minimal error, which explains the reduced gap towards TIN and SCM. Interestingly, observing M3's MSE performance, it is noteable that RSMA enables the mechanisms to closely approach the MSE performance before outage. In contrast, TIN and SCM experience a large gap between Init and M3, highlighting RSMA's resilience enhancing capabilities.
    
    
    \ifarxiv
	\subsubsection{Impact of Mixed Criticality}
	In this contribution, the concept of mixed criticality is utilized to express the diverse applications (and their physical layer demands) in a future 6G wireless communication networks. As the exact specifications of QoS targets within such system significantly impact the system performance, we herein assess the resilience and MSE as a function of the sum target rate, i.e., $\sum_{k\in\mathcal{K}}r_k^\mathrm{des}$. In accordance with the previous simulations, we respect the criticality levels, i.e., we start with $9$, $6$, $\SI{3}{Mbps}$ for HI, ME, and LO levels, respectively. For each point on the x-axis, we increment the LO level by $\SI{1}{Mbps}$ keeping the ME (HI) level at double (triple) the LO level. For reference, we implement four schemes with different approaches treating the criticality level: The first three schemes discard the difference of criticalities, i.e., the \emph{mixed}-aspect, setting all QoS demands to the same value, i.e., $(1)$ FixHI to HI, $(2)$ FixME to ME, and $(3)$ FixLO to LO criticality level; the fourth scheme discards the overall mixed criticality aspect, ignoring QoS demands, such that the algorithm corresponds to a sum rate maximization (SRM).
	
	For an outage probability of $30$\%, Fig.~\ref{fig:yRes_xSumQOS_M5_lambda_02_05_03} shows the resilience as a function of $\sum_{k\in\mathcal{K}}r_k^\mathrm{des}$, comparing the schemes. It is first noted that RSMA outperforms the references for all considered QoS targets, showing the superior resilience of the proposed schemes and emphasizing the idea of considering mixed criticality within such network. At low values for $\sum_{k\in\mathcal{K}}r_k^\mathrm{des}$, i.e., easy to fulfill QoS demands, FixHI performs better than FixME (and close to RSMA), as the algorithm can achieve most QoS targets.
	Between $84$ and $\SI{112}{Mbps}$, FixME outpaces FixHI, which is reasonable since FixME is a plausible compromise between setting all criticalities to LO or HI. However, this scheme does not beat the RSMA performance. FixLO performs worst, as the majority of users' demands are unfulfilled. Especially in the high QoS regime, SRM performs worst, as the consideration of maximizing the sum of all rates is not sufficient in terms of QoS demands and mixed criticality. Thus, these results show the importance of not only considering the criticality levels of the network, but also the mixed criticality aspect.
	
	A similar observation can be made from Fig.~\ref{fig:yMSE_xSumQOS_MI5}, which plots $\Psi$, i.e., the MSE, over $\sum_{k\in\mathcal{K}}r_k^\mathrm{des}$. SRM and FixLO are not able to support the network needs, as they result in high values for $\Psi$. FixHI, whilst achieving excellent MSE at low QoS demands, suffers severe losses in high demand regime. Thereby, these simulations highlight the explicit gain of the proposed mixed criticality-enabled RSMA paradigm in terms of minimizing rate gap in networks with high QoS demand, i.e., 6G networks.
    \fi
    \subsubsection{Longer Scale Resilience}
	In these set of simulations, we consider only the proposed scheme and investigate its performance as a function of time facing subsequent outage events, which increase in their severity over time. In more details, the first event denotes channel outages with $10$\% probability, and all following events have additional $10$\% probability each. The times of these outage events (upper plot), the throughput, which is defined as the sum of allocated rates that are achievable, the sum QoS target, and the adaption events (middle plot), as well as the momentarily resilience (lower plot) are shown in Fig.~\ref{ySR_yRes_xTime_RSMA_lambda_03_065_005v2}. Plenty of observations can be made:
	Algorithm~\ref{AlgCen} is able to recover the throughput completely (after some time) for events up to $70$\% channel outage probabilities. While the throughput is able to be restored (for events $1$-$6$), the resilience does never return to $1$, which comes from the time-to-recovery component. Regardless of the severity of outages, M1 offers outstanding time-to-recovery while also providing good adaption. Especially at the earlier events, consider the gap of resilience at the time of outage and at the time of M1's adaption: This gap is larger than at other mechanisms. While for lower outage probabilities, the first two mechanisms provide the best throughput adaption, and thus resilience, in high outage regime, M3 and M4 dominate the achievable adaption and the resilience quality. This results captures once again the trade-off between quality and time of adaption, which the proposed resiliency metric includes. Interestingly, even in the face of $90$\% outages, especially the \emph{BS-user-allocation adaption} is able to recover the throughput up to around $\SI{70}{Mbps}$, which is more than $50$\% of the QoS target. This results further emphasizes the suitability of C-RAN, especially due to the great amount of links between BSs and users, for resilient networks, as C-RAN provides great absorption and adaption potential.
    \begin{figure}
		\centering
		\includegraphics[width=\linewidth]{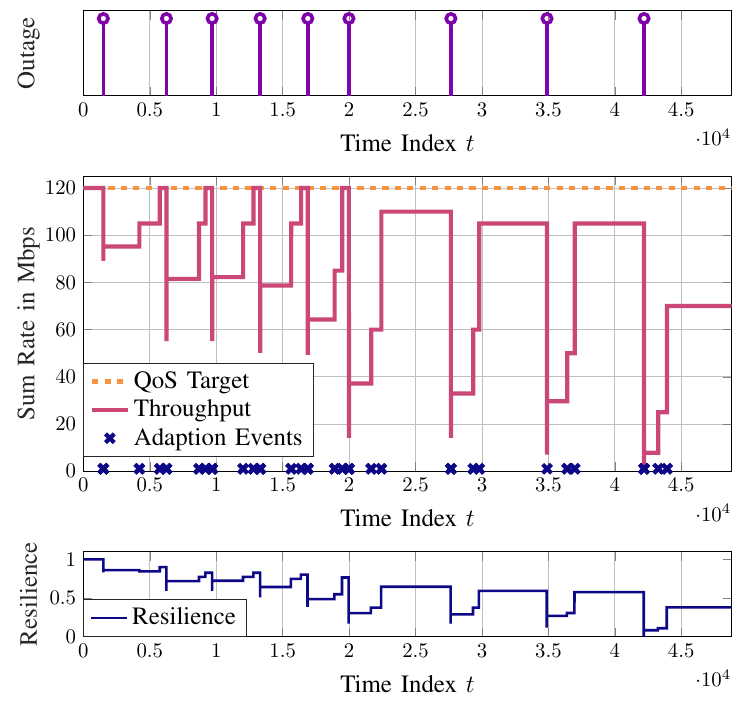}\vspace*{-.4cm}
		\caption{Throughput, adaption events, and resilience over the time index $t$ facing subsequent outage events.}
		\label{ySR_yRes_xTime_RSMA_lambda_03_065_005v2}
		\vspace*{-.75cm}
	\end{figure}%
	
    \subsubsection{Data Rate Performance} In the last set of simulations, for the ease of presentation, we consider a small network with $B=2$ BSs, $K=3$ users, and an outage probability of $60$\%, whereas users $1,2,3$ represent HI, ME, and LO critical levels, with $r_k^\mathrm{des}\in\{5,10,15\}\SI{}{Mbps}$. Fig.~\ref{zDisTime_yRTE_xUSR_RSMA_lambda_03_065_005_v2} depicts each users private and common data rate in Mbps for each discrete time point of the resilience procedure. In other words, Fig.~\ref{zDisTime_yRTE_xUSR_RSMA_lambda_03_065_005_v2} shows an example recovery process, similar to Fig.~\ref{ySR_yRes_xTime_RSMA_lambda_03_065_005v2} for a single outage event, but with more details in terms of user rates and RSMA message split. From the initial solution, which shows a balanced allocation of private and common rate for each user, the network experiences severe performance degradation after the outage. Users $1$ and $2$ degrade to zero-rate and user $3$ remains connected only by the common message. The fast M1 is able to recover parts of the rates, significantly boosting the performance back up, whereas mostly common rates can be restored. Consequently, M2 is able to restore the rate to a level similar to before the outage. At this point, Algorithm~\ref{AlgCen} may choose to abort the adaption process, dependent on the computational resources. However, by re-clustering, M3 changes the rate allocation towards enabling only private messages. This interesting behavior can be explained by the loss of about $60$\% of links, including interference links. In the new situation, common message decoding, as per the proposed RSMA scheme, is not needed due to the decreased interference. At last, M4 does not influence the resource allocation in this case. 
    These results, albeit missing the timing (time-to-recovery) information, show the high dynamics of private and common rate during the adaption process, which emphasize the promising factor that RSMA provides for high levels of resilience within a mixed-critical network.\vspace{-.1cm}
    \begin{figure}
		\centering
		\includegraphics[width=\linewidth]{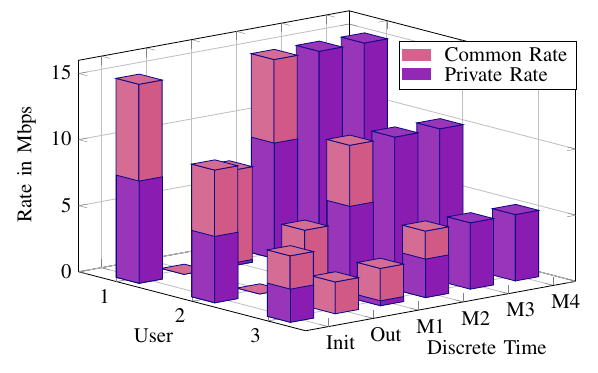}\vspace*{-.4cm}
		\caption{Rates over consequent discrete time points, namely, initial allocation, outage rate, recovered rates from M1-M4.}
		\label{zDisTime_yRTE_xUSR_RSMA_lambda_03_065_005_v2}
		\vspace*{-.6cm}
	\end{figure}%

	\section{Conclusion}\label{sec:con}\vspace*{-.1cm}
	Serving the needs of mixed-critical functional safety in future 6G networks is reliant on modern resource management techniques not only respecting the mixed criticality aspect but also providing high levels of resilience, to enable fault tolerance and ensuring safety for operating personnel, environment, and machinery. This contribution proposes a metric for jointly incorporating mixed criticality and resilience for the physical layer resource management respecting diverse QoS targets. This paradigm is investigated in a RSMA-enabled C-RAN subject to maximum transmit power, fronthaul capacity, and achievable private and common message rate constraints. An efficient resource management algorithm was derived using fractional programming, inner-convex approximation, and $l_0$-norm approximation. Building upon these considerations, four resiliency mechanisms differing in complexity and expected quality-of-adaption are proposed and combined into a resilience and criticality-aware RSMA-enabled resource management algorithm. Various numerical results show the dynamics of the proposed metric, the resiliency and MSE behavior, 
	\ifarxiv%
	impacts of different criticality schemes, 
	\fi%
	the resilience over consecutive outage events, and an explicit recovery procedure. Results show that there is a trade-off between the optimal resilience mechanism and the outage probability, RSMA outperforms reference interference management schemes, and the consideration of mixed criticality is a crucial factor in the considered system model. A simulation of the algorithm's behavior over time verifies the practical and numerical merits of the proposed scheme, and emphasizes its importance for future 6G wireless communication networks.
	\appendices
    \section{Proof of Theorem~\ref{th:1.convex}}\label{app3}
    \subsubsection{Inner-convex and $l_0$-norm approximation}
    Problem \eqref{eq:Opt1} is a mixed-integer non-convex optimization problem due to the variables $\bm{\mathcal{K}}$ and fronthaul constraint \eqref{eq:fthl}, reformulated as
    \begingroup\addtolength{\jot}{-.4cm}
    \begin{align}
        \sum_{{k \in \mathcal{K}}} \big(\norm[\big]{\norm[\big]{\bm{w}_{b,k}^p}_2^2}_0\, r_{k}^{p} + \norm[\big]{\norm[\big]{\bm{w}_{b,k}^c}_2^2}_0\, r_{k}^{c}\big) \leq C_b^{\mathrm{max}}, \quad \forall b\in\mathcal{B}. \nonumber\\
        \label{eq:fthl3}
    \end{align}\endgroup
    The discrete $l_0$-norm in \eqref{eq:fthl3} helps formulating the constraint in such a way, that a possible transmission of user $k$'s signal over the fronthaul link of BS $b$ is indicated by the entries of the precoding vector of $k$'s signal at $b$. In more details, $\norm[\big]{\norm[\big]{\bm{w}_{b,k}^p}_2^2}_0 = 1$, if and only if $b$ allocates some power towards $k$'s private message (similar for the common message, otherwise it is zero. Following \cite{6920005}, we approximate the $l_0$-norm with a weighted l$_1$-norm as $\norm[\big]{\norm[\big]{\bm{w}_{b,k}^o}_2^2}_0 \approx \beta_{b,k}^o\norm[\big]{\norm[\big]{\bm{w}_{b,k}^p}_2^2}_1$, where $\beta_{b,k}^o$ are the weights calculated by $\beta_{b,k}^o = \left(\delta+\norm[\big]{\bm{w}_{b,k}^o}_2^2\right)^{-1}$, with $\delta > 0$ being a regularization constant. Note that this formulation is an application of l$_1$-norm to a quadratic function of the precoders, which yields a smooth, convex, and continuous function. A few notes on the weights: In case BS $b$ assigns low transmit powers to user $k$'s private or common signal, $\beta_{b,k}^o$ increases. Having the choice to serve only few selected users with reasonable transmit power allocation, the algorithm would eventually exclude messages with high weights in order to achieve higher rates, since the fronthaul link is a bottleneck in C-RAN. This interplay naturally balances the load between the BSs and leaves users being served only by BSs with reasonable transmit power. The fronthaul constraint becomes thereby\vspace{-.2cm}
    \begingroup\addtolength{\jot}{-.4cm}
    \begin{align}
        \sum_{{k \in \mathcal{K}}}\big( \beta_{b,k}^p\norm[\big]{\bm{w}_{b,k}^p}_2^2\, r_{k}^{p} + \beta_{b,k}^c\norm[\big]{\bm{w}_{b,k}^c}_2^2\, r_{k}^{c}\big) \leq C_b^{\mathrm{max}}, \quad \forall b\in\mathcal{B}. \nonumber\\\label{eq:fthl4}
    \end{align}\endgroup
    \\\vspace{-1.1cm}\\
    Note that the l$_1$-norm is omitted here as the argument is scalar. By introducing the slack variable $\bm{u}$, we transform \eqref{eq:fthl4} into the constraints \eqref{eq:betacon} and\vspace{-.2cm}
    \begin{equation}
        \sum_{{k \in \mathcal{K}}} \big(u_{b,k}^p\, r_{k}^{p} + u_{b,k}^c\, r_{k}^{c}\big) \leq C_b^{\mathrm{max}}, \quad \forall b\in\mathcal{B}, \label{eq:fthl5}\vspace{-.2cm}
    \end{equation}
    which is bilinear in the optimization variables and thereby amenable for applying ICA. In this context, we first provide an equivalent difference of convex formulation of \eqref{eq:fthl5} and subsequently create a first-order Taylor series expansion of the non-convex terms around operating points:\vspace{-.3cm}
    \begingroup\addtolength{\jot}{-.4cm}
    \begin{align}
        &\sum_{{k \in \mathcal{K}}} \frac{1}{4}\big(\big((u_{b,k}^p + r_{k}^{p})^2-(u_{b,k}^p - r_{k}^{p})^2\big) \nonumber\\
        &\qquad\qquad+ \big((u_{b,k}^c + r_{k}^{c})^2-(u_{b,k}^c - r_{k}^{c})^2\big)\big) \leq C_b^{\mathrm{max}},\label{eq:fthl6}
    \end{align}\endgroup
    \\\vspace{-1.1cm}\\
    and the first-order Taylor series expansion around $\tilde{\bm{u}}$ and $\tilde{\bm{r}}$, being feasible fixed values from the previous iteration's solution, is then formulated into \eqref{eq:fthl2}. Note that $\tilde{\bm{u}}$ and $\tilde{\bm{r}}$ are vectors of similar dimension as ${\bm{u}}$ and ${\bm{r}}$, respectively.
	\subsubsection{Quadratic transform}
	Due to the complex fractional form of constraints \eqref{eq:achp} and \eqref{eq:achc}, the problem \eqref{eq:Opt1} is of non-convex nature. We consider the following Lemma~\ref{lma_1}, to introduce $\gamma_k^p$ and $\gamma_{i,k}^c$ for the SINR terms, and transform the original non-convex problem.
	\begin{lemma}\label{lma_1}
	    A rewritten formulation of the optimization problem \eqref{eq:Opt1}, including ICA and $l_0$-norm relaxation, is given by\vspace{-.1cm}%
    	\begin{subequations}\label{eq:Opt3}
        \begingroup
        \addtolength{\jot}{-.1cm}
    		\begin{align}
    			\underset{\bm{w},\bm{r},\bm{{u}},\bm{\gamma}}{\mathrm{min}}\quad &\Psi  \tag{\ref{eq:Opt3}} \\
    			\mathrm{s.t.} \quad &\eqref{eq:pwr}, \eqref{eq:fthl2}, \eqref{eq:betacon},\nonumber\\
	    & r_{k}^{p} \leq \tau\,\log_2(1+\gamma_{k}^p),  &\forall k &\in \mathcal{K}, \label{eq:achp3}\\	
		& r_{i}^{c} \leq \tau\,\log_2(1+\gamma_{i,k}^c), &\forall i\in\mathcal{I}_k, \forall k &\in \mathcal{K}, \label{eq:achc3}\\
	    & \gamma_{k}^p \leq \Gamma_{k}^p,  &\forall k &\in \mathcal{K}, \label{eq:sinrp}\\	
		& \gamma_{i,k}^c \leq \Gamma_{i,k}^c, &\forall i\in\mathcal{I}_k, \forall k &\in \mathcal{K}. \label{eq:sinrc}
    		\end{align}\endgroup
    	\end{subequations} 
	    Given the stationary solution $(\bm{w}^\star,\bm{r}^\star,\bm{u}^\star,\bm{\gamma}^\star)$ to problem \eqref{eq:Opt3}, we note that $(\bm{w}^\star,\bm{r}^\star,\bm{u}^\star)$ is a stationary solution of the l$_1$-norm approximation of problem \eqref{eq:Opt1}.
	\end{lemma}
	\vspace{-.2cm}
	\ifarxiv
	\begin{proof}
	    Both problems \eqref{eq:Opt1} and \eqref{eq:Opt3} share the same objective and constraints \eqref{eq:achp3} and \eqref{eq:sinrp} can be formulated as 
	    \begin{equation}
	        r_{k}^{p} \leq \tau\,\log_2(1+\gamma_{k}^p) \leq \tau\,\log_2(1+\Gamma_{k}^p).
	    \end{equation}
	    When engineering the newly introduced optimization variable $\gamma_k^p$ to fulfill $r_{k}^{p} = \tau\,\log_2(1+\gamma_{k}^p)$ and $\gamma_{k}^p = \Gamma_{k}^p$, we get \eqref{eq:achp}. A similar statement holds for the common messages, which completes the proof.
	\end{proof}\fi
	Problem \eqref{eq:Opt3} is still non-convex due to the constraints \eqref{eq:sinrp} and \eqref{eq:sinrc}. However, such formulations are amenable for applying fractional programming techniques. Especially, we utilize the quadratic transform (QT) in multidimensional and complex case proposed by \cite[Theorem 2]{8314727}, tailored to the constraints \eqref{eq:sinrp} and \eqref{eq:sinrc} in this work. In more details, after subtracting the right term from both sides and applying QT, the functions can be formulated as given in \eqref{eq:qtp} and \eqref{eq:qtc}, respectively.
	In this context, \eqref{eq:qtp} and \eqref{eq:qtc} are the QT formulations of \eqref{eq:sinrp} and \eqref{eq:sinrc}, respectively, where ${\chi}_{k}^p$ and ${\chi}_{i,k}^c$ are auxiliary variables. Consider following remark on the convexity of \eqref{eq:qtp} and \eqref{eq:qtc}, and Lemma~\ref{lem:u} for obtaining the optimal auxiliary variables, when $\bm{w}$ and $\bm{\gamma}$ are fixed.\vspace{-.2cm}
	\begin{remark}
	    For fixed ${\chi}_{k}^p$ and ${\chi}_{i,k}^c$, the second terms in \eqref{eq:qtp} and \eqref{eq:qtc} become linear functions of the precoders, also the latter terms become convex. Thus, \eqref{eq:qtp} and \eqref{eq:qtc} denote convex functions of the procoding vectors and the SINR variables, in case the auxiliary variables are fixed.
	\end{remark}\vspace{-.2cm}
	\begin{lemma}\label{lem:u}
	    The optimal auxiliary variable results are\vspace{-.2cm}
	    \begingroup\addtolength{\jot}{-.2cm}
	\begin{align}
	    &{\chi}_{k}^p = \frac{\left(\bm{w}_{k}^p\right)^H \bm{h}_{k}}{\sigma^2 + \sum\limits_{j \in \mathcal{K}\setminus \{k\}}\left|\bm{h}_{k}^{H}\bm{w}_{j}^p \right|^2 + \sum\limits_{l \in \mathcal{K}\setminus \mathcal{I}_k}\left|\bm{h}_{k}^{H}\bm{w}_{l}^c \right|^2},\label{eq:ukp}\\
	    &{\chi}_{i,k}^c = \frac{\left(\bm{w}_{i}^c\right)^H \bm{h}_{k}}{\sigma^2 + \sum\limits_{j \in \mathcal{K}}\left|\bm{h}_{k}^{H}\bm{w}_{j}^p \right|^2 + \sum\limits_{l \in \mathcal{K}\setminus \mathcal{I}_k\cup \mathcal{I}'_{k,i}}\left|\bm{h}_{k}^{H}\bm{w}_{l}^c \right|^2}.\label{eq:ukc}
	\end{align}\endgroup
	\end{lemma}
	\vspace{-.2cm}
	\ifarxiv
	\begin{proof}
	    In \eqref{eq:qtp} and \eqref{eq:qtc}, the partial derivatives of $g^p(\bm{w},\bm{\gamma})$ and $g^c(\bm{w},\bm{\gamma})$ with respect to ${\chi}_{k}^p$ and ${\chi}_{i,k}^c$ are set to zero and solved, respectively. From this follow the optimal auxiliary variable results in \eqref{eq:ukp} and \eqref{eq:ukc}.
	\end{proof}\fi
	Due to the nature of the QT, which requires the update of auxiliary variables in an iterative fashion, the overall solution results an iterative procedure. Thereby, the algorithm computes the auxiliary variables following Lemma~\ref{lem:u}, for given beamforming vectors, and consequently solves a convex problem resulting optimal resource allocation variables.
	At each iteration, the optimization problem is given in \eqref{eq:Opt2} (as defined above).
	Hereby, the objective function \eqref{eq:Opt2} and the feasible set defined by all constraints are convex. Therefore, problem \eqref{eq:Opt2} is a convex optimization problem that can efficiently solved using established solvers, such as CVX \cite{cvx}. 
	\ifarxiv
	\section{Solution to Problem~\eqref{eq:Opt2}}\label{appB}
	Similar to the feasible fixed values $\tilde{\bm{u}}$ and $\tilde{\bm{r}}$, which originate from the inner-convex approximation, $\beta_{b,k}^o$ and ${\chi}_{k}^o$ need to be updated iteratively. Thereby, to solve the resource management problem efficiently, a procedure that solves the convex problem \eqref{eq:Opt2} and updates all auxiliary variables, weights, and fixed values iteratively is implemented in Algorithm \ref{alg}.
	\begin{algorithm}[h]
    \caption{RSMA-enabled Resource Management for Rate Gap Minimization}
    \begin{algorithmic}[1]
	\STATE Initialize feasible precoders $\bm{w}$\vspace*{-.05cm}
    \STATE Determine RSMA-related sets $\mathcal{M}_k$, $\mathcal{I}_k$, $\mathcal{I}'_{i,k}$, $\pi_k$\vspace*{-.05cm}
	\REPEAT
	\STATE Update the auxiliary variables ${\chi}_{k}^p$ and ${\chi}_{i,k}^c$, the weights $\beta_{b,k}^o$, and the feasible fixed values $\tilde{u}_{b,k}^o$ and $\tilde{r}_k^o$ \vspace*{-.05cm}
	\STATE Solve the convex problem \eqref{eq:Opt2} 
	\UNTIL{convergence}
    \end{algorithmic}
    \label{alg}
    \end{algorithm}%
	
	In Algorithm~\ref{alg}, detailed steps for beamformer design and rate allocation are provided. Initially, the beamformers are computed to feasible values. This can be done using a \emph{random} initialization or by computing the maximum ratio transmitters. Subsequently, RSMA-related sets, i.e., $\mathcal{M}_k$, $\mathcal{I}_k$, $\mathcal{I}'_{i,k}$, and $\pi_k$, and thereafter, the clustering sets for private and common streams $\mathcal{K}_b^p$ and $\mathcal{K}_b^p$, are computed. 
	Two steps are repeated until convergence: $(a)$ The auxiliary variables are updated according to \eqref{eq:ukp} and \eqref{eq:ukc}; $(b)$ Problem \eqref{eq:Opt2} is solved using CVX. Consider the following Lemma~\ref{lma_3} on the convergence of Algorithm~\ref{alg}.
    
    \begin{lemma}\label{lma_3}
        The iterative procedure in Algorithm~\ref{alg} yields the stationary solution $(\bm{w}^\star,\bm{r}^\star,\bm{u}^\star)$ to the l$_1$-norm relaxed version of problem \eqref{eq:Opt1}.
    \end{lemma}
    \begin{proof}
        With given Lemma \ref{lma_1}, the proof remains to show that the iterative procedure yields $(\bm{w}^\star,\bm{r}^\star,\bm{u}^\star,\bm{\gamma}^\star)$, a stationary solution to problem \eqref{eq:Opt3}. The steps in this proof are similar to the proof of \cite[Theorem 3]{8314727} and rely on the therein provided conditions of equivalent solution and equivalent objective. Algorithm \ref{alg} is a block coordinate descent algorithm for problem \eqref{eq:Opt2}, which is a convex problem. Therefore, such algorithm is guaranteed to converge to a stationary solution of problem \eqref{eq:Opt2}. Due to the definition of the auxiliary variables, equations \eqref{eq:qtp} and \eqref{eq:qtc} yield \eqref{eq:sinrp} and \eqref{eq:sinrc}, respectively, if and only if $({u}_{k}^p)^\star$ and $({u}_{i,k}^c)^\star$ are the optimal auxiliary variables of the stationary solution to \eqref{eq:Opt2}. As problems \eqref{eq:Opt3} and \eqref{eq:Opt2} share the same objective, and constraints \eqref{eq:qtp} and \eqref{eq:qtc} yield \eqref{eq:sinrp} and \eqref{eq:sinrc}, respectively, Algorithm \ref{alg} also converges to a stationary solution to problem \eqref{eq:Opt3}. This completes the proof.
    \end{proof}
	
	The overall computational complexity of Algorithm~\ref{alg} depends on the interior-point method to solve problem \eqref{eq:Opt2}, whereas the upper-bound computational complexity is $\mathcal{O}({V}_{\text{max}}(d_1)^{3.5})$. Note that ${V}_{\text{max}}$ is the number of iterations until convergence for the worst-case and $d_1={K}(2{B}({L}+1)+{K}+3)$ is the total number of variables. Most resilience mechanisms depend on solving relaxed versions of problem \eqref{eq:Opt2}, and differ in the number of variables. The complexity can be derived in a similar manner.

    \section{Generalized Assignment Problem-based Clustering}\label{app4}
    \label{GAP}
	The clustering sets are computed by a generalized assignment problem (GAP) formulation. To obtain $\mathcal{K}_b^p$ and $\mathcal{K}_b^c$, we define the binary variable $\nu_{b,k}^o\in\{0,1\}$, referring to BS $b$ serving the private (or common) message of user $k$. Next, we give a GAP-based formulation which captures the benefit of assigning BS $b$ to serve a message intended for user $k$. Such benefit is defined in terms of the channel norm, so as to preferably utilize strong links. This can be mathematically formulated as
	\begin{subequations}\label{eq:GAP}
		\begin{align}
			\underset{\bm{\nu}}{\mathrm{max}}\quad &\sum_{(k,b)\in(\mathcal{K},\mathcal{B})} \left( \nu_{b,k}^p \norm[\big]{\bm{h}_{b ,k}}_2^2 + \nu_{b,k}^c \sum_{i\in\mathcal{M}_k} \norm[\big]{\bm{h}_{b ,i}}_2^2 \right)\hspace{-3.5cm}&&\hspace{2cm} \tag{\ref{eq:GAP}} \\
		    \mathrm{s.t.} \quad &\sum_{b\in\mathcal{B}} \nu_{b,k}^o \leq B_k^\mathrm{max} &\forall k\in\mathcal{K}, \forall o&\in\{p,c\} \label{eq:nmax},\\
	        &\sum_{k\in\mathcal{K}} \nu_{b,k}^p + \nu_{b,k}^c \leq I_b^\mathrm{max} &\forall b&\in\mathcal{B} \label{eq:kmax},\\
	        &\nu_{b,k}^p + \nu_{b,k}^c \leq 1 &\forall k\in\mathcal{K},\forall b&\in\mathcal{B} \label{eq:diffbs}.
	   \end{align}
	\end{subequations} 
	Problem \eqref{eq:GAP} maximizes a channel quality utility by jointly optimizing the binary clustering variables $\bm{\nu}=\mathrm{vec}( \nu_{b,k}^o | \forall (b,k)\in(\mathcal{B},\mathcal{K}), \forall o\in\{p,c\} )$ and is in the form of a integer linerar program (ILP). Constraint \eqref{eq:nmax} restricts the maximum number of BSs that serve the private (common) message of each user, where $B_k^\mathrm{max}$ are the maximum number of BSs per message. Each BS can only serve a fixed number of messages, which is denoted by \eqref{eq:kmax}, where $I_b^\mathrm{max}$ describes the maximal amount of supported messages. Both of the previously mentioned constraints help balancing the load. Constraint \eqref{eq:diffbs} is especially chosen to enhance the resilience behavior of the proposed RSMA scheme. The idea behind \eqref{eq:diffbs} is to split the serving BSs of both private and common message of user $k$. Thereby, the scheme behaves more resilient when outages occur. 
	
	Problem \eqref{eq:GAP} follows a GAP structure \cite{6697043} and can be solved using well studied methods, e.g., branch and cut algorithm \cite{6567888}. This allows us to fix the clustering sets $\mathcal{K}_b^p$ and $\mathcal{K}_b^c$, by setting $\mathcal{K}_b^o = \{ k\in\mathcal{K}|\nu_{b,k}^o = 1 \}$. 
	\fi
    \bibliographystyle{IEEEtran}
	\bibliography{bibliography}
\end{document}